\documentclass[journal]{IEEEtran}
\ifCLASSINFOpdf
\else
   \usepackage{graphicx}
   \graphicspath{{../eps/}}
   \DeclareGraphicsExtensions{.eps}
\fi

\usepackage{amsmath}
\usepackage{amsfonts,amssymb}
\usepackage{amssymb}
\usepackage{wrapfig}
\usepackage{psfrag}
\usepackage{epstopdf}
\usepackage{cite}
\usepackage{graphicx}
\usepackage{subfigure}
\usepackage{threeparttable}
\usepackage{cases}
\usepackage{subeqnarray}
\usepackage{color}
\usepackage{underscore}
\usepackage{verbatim}
\usepackage{bm}
\usepackage{stfloats}
\usepackage{xpatch}
\usepackage{makecell}

\newtheorem{lemma}{Lemma}

\newtheorem{algorithm}{Algorithm}
\newtheorem{proposition}{Proposition}

\usepackage{algorithm}
\usepackage{algorithmic}
\usepackage[table]{xcolor}

\begin{document}
\title{UAV Swarm Enabled Aerial Movable Antenna System for Low-Altitude Economy: From\\ Far-Field to Near-Field Communication
}
%
%
%
\author{Haiquan~Lu,
        Chao~Feng,
        Yong~Zeng,~\IEEEmembership{Fellow,~IEEE,}
        Shaodan~Ma,~\IEEEmembership{Senior Member,~IEEE,}\\
        Long~Shi,~\IEEEmembership{Senior Member,~IEEE,}
        Shi~Jin,~\IEEEmembership{Fellow,~IEEE,}
        and
        Rui~Zhang,~\IEEEmembership{Fellow,~IEEE}
\thanks{Haiquan Lu is with the School of Electronic and Optical Engineering, Nanjing University of Science and Technology, Nanjing 210094, China. He was with the State Key Laboratory of Internet of Things for Smart City and the Department of Electrical and Computer Engineering, University of Macau, Macao SAR, China, and also with the National Mobile Communications Research Laboratory, Southeast University, Nanjing 210096, China (e-mail: haiquanlu@njust.edu.cn).}
\thanks{Chao Feng is with the Department of Information Engineering, the Chinese University of Hong Kong, Hong Kong (e-mail: fc025@ie.cuhk.edu.hk).}
\thanks{Yong Zeng is with the National Mobile Communications Research Laboratory, Southeast University, Nanjing 210096, China, and is also with the Purple Mountain Laboratories, Nanjing 211111, China (e-mail: yong_zeng@seu.edu.cn). }
\thanks{Shaodan Ma is with the State Key Laboratory of Internet of Things for Smart City and the Department of Electrical and Computer Engineering, University of Macau, Macao SAR, China (e-mail: shaodanma@um.edu.mo).}
\thanks{Long Shi is with the School of Electronic and Optical Engineering, Nanjing University of Science and Technology, Nanjing 210094, China (e-mail: longshi@njust.edu.cn).}
\thanks{Shi Jin is with the School of information Science and Engineering, Southeast University, Nanjing 210096, China (e-mail: jinshi@seu.edu.cn).} 
\thanks{Rui Zhang is with the Department of Electrical and Computer Engineering, National University of Singapore, Singapore 117583 (e-mail: elezhang@nus.edu.sg).
}
}

\maketitle

 \begin{abstract}
 Unmanned aerial vehicle (UAV) with the intrinsic three-dimensional (3D) mobility provides an ideal platform for implementing aerial movable antenna (AMA) system enabled by UAV swarm cooperation. Besides, AMA system is readily to achieve an extremely large-scale array aperture, rendering the conventional far-field uniform plane wave (UPW) model no longer valid for aerial-to-ground links. This paper studies the UAV swarm enabled near-field AMA communication, by taking into account the non-uniform spherical wave (NUSW) model, where UAV swarm trajectory simultaneously influences the channel amplitude and phase. We formulate a general optimization problem to maximize the minimum average communication rate over user equipments (UEs), by jointly optimizing the 3D UAV swarm trajectory and receive beamforming for all UEs. To draw useful insights, the special case of single UE is first studied, and successive convex approximation (SCA) technique is proposed to efficiently optimize the UAV swarm trajectory. For the special case of placement optimization, the optimal  placement positions of UAVs for cases of single UAV and two UAVs are derived in closed-form. Then, for the special case of two UEs, we show that an inter-UE interference (IUI)-free communication can be achieved by symmetrically placing an even number of UAVs along a hyperbola, with its foci corresponding to the locations of the two UEs. Furthermore, for arbitrary number of UEs, an alternating optimization algorithm is proposed to efficiently tackle the non-convex optimization problem. Numerical results validate the significant performance gains over the benchmark schemes.
 \end{abstract}


\IEEEpeerreviewmaketitle
\section{Introduction}
 Multiple-input multiple-output (MIMO) has achieved great success in the fourth-generation (4G) mobile networks, owing to its improved diversity and multiplexing gains over single-input single-output (SISO) systems. In current fifth-generation (5G) networks, MIMO has been evolved into massive MIMO, which is anticipated to be further advanced to extremely large-scale MIMO (XL-MIMO) in the forthcoming sixth-generation (6G) era. In XL-MIMO systems, the number of antennas at the base station (BS) is drastically scaled up to hundreds or even thousands \cite{tong20216g,lu2024tutorial}, thus enabling unprecedented improvements in spectral efficiency and spatial resolution. Nevertheless, such improvements may come at the expense of significantly increased hardware cost and energy consumption, if each antenna entails a dedicated radio-frequency (RF) chain in digital beamforming \cite{heath2018foundations}. To address the above issues, several approaches have been proposed, such as activating less number of antennas with favorable channel conditions via antenna selection (AS), pre-configuring antennas in a sparse manner for achieving a larger physical and/or virtual aperture, i.e., sparse array \cite{wang2024enhancing,li2025sparse}, as well as adopting the hybrid beamforming architecture based on phase shifters \cite{el2014spatially,sohrabi2016hybrid}, lens arrays, and switches. Moreover, a novel cost-effective array architecture termed ray antenna array (RAA) was proposed in \cite{dong2025novel}, which contains multiple ray-like simple uniform linear arrays (sULAs) with properly designed orientations, providing enhanced beamforming gains without the need of phase shifters.

 However, the aforementioned antenna arrays typically employ the fixed-position antennas (FPAs). Recently, movable antenna (MA) \cite{zhu2024movable,zhu2025tutorial} or fluid antenna system (FAS) \cite{wong2021fluid} emerges as a promising technique to unlock the potential in antenna reconfigurability, by endowing the antenna with the capability of free movement through mechanical or electronic control. The similar idea of adjusting antenna positions can also be found in pinching-antenna systems \cite{ding2025flexible}. In particular, favorable channel conditions can be achieved by flexibly adjusting antenna positions, thus facilitating communication performance in terms of spectral efficiency and outage probability, as well as sensing performance with respect to (w.r.t.) sensing signal-to-noise ratio (SNR) and Cram\'{e}r-Rao bound (CRB) \cite{zhu2025tutorial}, without increasing the number of antennas. Furthermore, MA was generalized to six-dimensional movable antenna (6DMA) architecture in \cite{shao20256dModeling}, which holistically incorporates the design degrees of freedom (DoFs) of three-dimensional (3D) positions and 3D rotations. Such an ultimate antenna reconfigurability facilitates a better exploitation of channel spatial variations, and an enhanced system performance can be achieved compared to existing MA systems, especially when utilizing directional antennas \cite{shao20256ddiscrete,shao2025distributed}.

 It is worth noting that in terms of flexible movement, unmanned aerial vehicle (UAV) provides an ideal platform for 3D mobility, which has found assorted applications in low-altitude economy \cite{song2025overview,jiang2025integrated,he2025ubiquitous}, including integrated sensing and communication (ISAC), logistics and delivery, infrastructure inspection, emergency rescue, as well as agricultural and forestry plant protection, thanks to its inherent benefits of 3D maneuverability, flexible deployment, cost-effectiveness, and high line-of-sight (LoS) probability \cite{zeng2019accessing,mozaffari2019tutorial}.
 On one hand, UAV-assisted communication extends the conventional two-dimensional (2D) terrestrial networks to 3D domain, enabling flexible communication services from the sky with UAVs functioning as aerial BSs, relays, and access points (APs). On the other hand, similar to terrestrial user equipments (UEs), UAVs can also function as aerial UEs that access cellular networks, giving rise to a new paradigm termed cellular-connected UAVs \cite{zeng2019accessing}. In this context, extensive research efforts have been devoted to UAV communications in terms of various performance metrics, such as energy efficiency \cite{zeng2017energy,zeng2019energy}, achievable rate \cite{wu2018joint,wang2024learning}, outage probability \cite{mozaffari2016unmanned,lyu2019network,zeng2021simultaneous}, mission completion time \cite{zhang2019cellular,xu2021completion}, and secure rate \cite{zhong2018secure,wen2024ris}. 

 By leveraging the design DoF of UAV placement/trajectory, there are some preliminary efforts devoted to the integration of UAV and MA \cite{liu2025uav,ren20256d,kuang2024movable,zhou2025movable}. For example, the authors in \cite{liu2025uav} proposed mounting the MA array on a single UAV, so as to maximize the sum rate in multi-user communication system. In \cite{ren20256d}, the 6DMA was proposed to be mounted on a single cellular-connected UAV for enhanced co-channel interference mitigation. In \cite{kuang2024movable}, a UAV equipped with the MA array was applied in an ISAC system to maximize the achievable communication rate, while satisfying the sensing requirement. Note that existing works mainly concentrate on UAV-mounted MA arrays, whereas the inherent 3D mobility of UAVs render them well-suited for constructing MA systems. Specifically, by equipping each UAV with either a single antenna or multiple antennas, a MA system is enabled through the synergistic cooperation of the UAV swarm. Thus, flexible array architectures can be realized by adjusting the UAV positions, obviating the hardware components required for mechanical or electronic control in conventional MA systems. Moreover, compared to the conventional MA system, UAV swarm enabled aerial MA (AMA) system can achieve a larger-scale spatial movement for exploiting channel variations, and exhibit an enhanced scalability thanks to the benefit of flexible deployment. In \cite{lu2025UAVMA}, a novel UAV swarm enabled two-level AMA system was proposed, where each UAV is equipped with an individual MA array, and the UAV swarm further collaboratively forms a large-scale MA system, thus enabling a dual-scale antenna movement for enhancing the UAV-terrestrial communication performance. However, the work \cite{lu2025UAVMA} neglected the near-field impact of UAV placement variation on the channel amplitude, and did not leverage the benefits of UAV swarm trajectory optimization.

 Furthermore, UAV swarm enabled AMA system can readily realize an array aperture comparable to or even larger than that of XL-MIMO, under which UEs and/or scatterers are highly likely to be located in its near-field region. As such, the more general near-field non-uniform spherical wave (NUSW) model should be taken into account for the aerial-to-ground channels \cite{lu2024tutorial}, rather than the conventional far-field uniform plane wave (UPW) model. Specifically, with the NUSW model, the angles of arrival/departure (AoAs/AoDs) for signals between UE and BS antennas are no longer identical, and each antenna in general experiences different signal amplitude from any UE, which has spurred much research interest in near-field performance analysis and system optimization design \cite{lu2024tutorial,Han2023Towards}. Thus, the UAV swarm placement/trajectory influences not only the channel amplitude, but also the channel phase of aerial-to-ground links, enabling the simultaneous adjustment of both the large-scale path loss and small-scale fading. This differs from existing MA systems that only adapt to small-scale channel variations. In this paper, we consider a UAV swarm enabled near-field AMA communication system, where the UAV swarm cooperatively serves multiple UEs. Our main contributions are summarized as follows.
 \begin{itemize}[\IEEEsetlabelwidth{12)}]
 \item First, by considering the general NUSW model, we present the near-field channel from each UE to the UAV swarm enabled AMA system, which characterizes the impact of UAV swarm trajectory on both the channel amplitude and phase of aerial-to-ground links. Then, under the practical constraints on maximum speed and collision avoidance among UAVs, as well as initial and final positions constraints, a general optimization problem is formulated to maximize the minimum average communication rate over all the UEs, by jointly optimizing the 3D UAV swarm trajectory and receive beamforming for all the UEs at the UAV swarm.
 \item Second, to draw useful insights, we first consider the special case of single UE, where the UAV swarm trajectory can be efficiently optimized with the successive convex approximation (SCA) technique \cite{zeng2019accessing}. When the number of time slots equals one or is sufficiently large such that the intermediate movement time slots are negligible, the optimal placement positions of UAVs for cases of single UAV and two UAVs with equal altitudes are further derived in closed-form. Moreover, for the special case of two UEs, a UAV swarm placement scheme that achieves the inter-UE interference (IUI)-free communication is derived for even number of UAVs. The result shows that the UAVs should be symmetrically placed on a hyperbola, with the foci corresponding to the locations of two UEs.
 \item Last, for arbitrary number of UEs, an alternating optimization algorithm is proposed to tackle the non-convex optimization problem sub-optimally. By exploiting the fact that the channel amplitude is less sensitive to UAV swarm trajectory variation than the channel phase, we first obtain an initial UAV swarm trajectory by focusing on the channel amplitude, which aims to maximize the minimum average SNR over UEs. Then, with the channel amplitude fixed by the obtained trajectory, the receive beamforming and UAV swarm trajectory are alternately optimized. Numerical results are provided to show the significant performance gain of the proposed scheme over the benchmark schemes.
 \end{itemize}

 The rest of the paper is organized as follows. Section~\ref{sectionSystemModel} presents the system model of UAV swarm enabled near-field AMA communication, and formulates the problem for maximizing the minimum average communication rate over UEs. Section~\ref{sectionSingleTWOUE} studies the special cases of single UE and two UEs, respectively. Section~\ref{sectionArbitraryUE} generalizes the results to arbitrary number of UEs. Section~\ref{sectionNumericalResult} provides the numerical results to evaluate the performance of UAV swarm enabled near-field AMA system. Finally, Section~\ref{sectionConclusion} concludes the paper.

 \emph{Notations:} Scalars are denoted by italic letters. Vectors and matrices are denoted by bold-face lower- and upper-case letters, respectively. ${{\mathbb{C}}^{M \times N}}$ denotes the space of $M \times N$ complex-valued matrices. $\left\| {\bf{x}} \right\|$ denotes the Euclidean norm of a vector ${\bf{x}}$. For an arbitrary-size matrix ${\bf A}$, its complex conjugate, transpose, and Hermitian transpose are denoted by ${\bf A}^*$, ${\bf A}^T$, ${\bf A}^H$, respectively. The spectral norm and Frobenius norm of ${\bf A}$ are denoted as ${\left\| {\bf{A}} \right\|_2}$ and ${\left\| {\bf{A}} \right\|_F}$, respectively. The distribution of a circularly symmetric complex Gaussian (CSCG) random vector with mean $\bf{x}$ and covariance matrix $\bf{\Sigma}$ is denoted as ${\cal CN}\left( {\bf{x},\bf{{\Sigma}}} \right)$, and $\sim$ stands for ``distributed as". The symbol ${\rm j}$ denotes the imaginary unit of complex numbers, with ${{\rm j}^2} =  - 1$. For a complex-valued number $x$, ${\mathop{\rm Re}\nolimits} \left\{ x \right\}$ denotes its real part. 
 ${\cal O}\left({\cdot}\right)$ denotes the standard big-O notation.
\section{System Model And Problem Formulation}\label{sectionSystemModel}
 \begin{figure}[!t]
 \centering
 \centerline{\includegraphics[width=3.5in,height=2.4in]{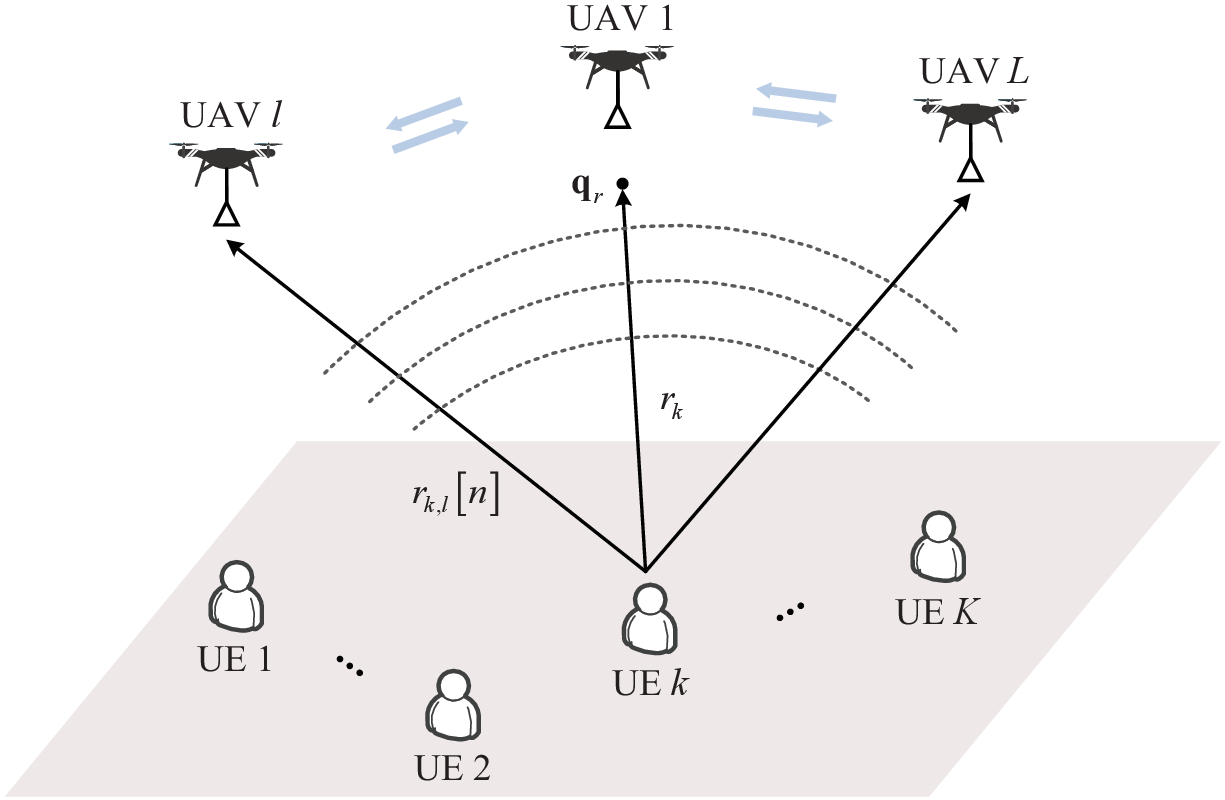}}
 \caption{A low-altitude UAV swarm enabled near-field AMA communication system.}
 \label{fig:systemModel}
 \end{figure}

 As shown in Fig.~\ref{fig:systemModel}, we consider a near-field wireless communication system enabled by a low-altitude UAV swarm, which consists of $L$ single-antenna UAVs and $K$ single-antenna ground UEs. With a 3D Cartesian coordinate system, the location of ground UE $k$, $k \in {\cal K}\triangleq \left\{{1, \cdots, K}\right\}$, is denoted as ${{\bf{w}}_k} = {\left[ {{x_k},{y_k},0} \right]^T}$.
 Let ${{\bf{q}}_r} = {\left[ {{x_r},{y_r},{z_r}} \right]^T}$ denote the reference point for the $L$ UAVs, and the distance between ${{\bf{q}}_r}$ and ${{\bf{w}}_k}$ is ${r_k} = \left\| {{{\bf{q}}_r} - {{\bf{w}}_k}} \right\|$. The 3D trajectory of UAV $l$ is denoted as ${{\bf{q}}_l}\left( t \right) = {\left[ {{x_l}\left( t \right),{y_l}\left( t \right),{z_l}\left( t \right)   } \right]^T}$, where $ 0 \le t \le T$, with $T$ being the operation time horizon. Moreover, by discretizing $T$ into $N$ time slots, i.e., $T = N{\delta _t}$, where $\delta _t$ is a sufficiently small time slot length, the position of UAV $l$ at time slot $n$ is ${{\bf{q}}_{l}}\left[ n \right] = {\left[ {{x_{l}}\left[ n \right],{y_{l}}\left[ n \right],{z_{l}}\left[ n \right]} \right]^T} $, $ n \in {\cal N}$, with ${\cal N} \triangleq \left[1, \cdots, N\right]$. In practice, there exist maximum speed and collision avoidance constraints for UAV swarm \cite{zeng2019accessing}, which are respectively given by
 \begin{equation}\label{speedConstraint}
 \left\| {{{\bf{q}}_l}\left[ {n + 1} \right] - {{\bf{q}}_l}\left[ n \right]} \right\| \le {{\tilde V}_{\max }},\ \forall l\ {\rm and}\  n = 1, \cdots, N-1,
 \end{equation}
 \begin{equation}\label{collisionAvoidanceConstraint}
 \left\| {{{\bf{q}}_l}\left[ n \right] - {{\bf{q}}_{l'}}\left[ n \right]} \right\| \ge {d_{\min }},\ \forall n,l, l' \ne l,
 \end{equation}
 where ${{\tilde V}_{\max }} \triangleq {V_{\max }}{\delta _t}$, with $V_{\max}$ being the maximum UAV speed in meter/second (m/s), and $d_{\min}$ denotes the minimum distance to avoid the collision between UAVs. Moreover, the distance between UE $k$ and UAV $l$ at slot $n$ is
 \begin{equation}\label{distanceUEkUAVl}
 \begin{aligned}
 {r_{k,l}}\left[ n \right]& = \left\| {{{\bf{q}}_l}\left[ n \right] - {{\bf{w}}_k}} \right\|\\
 & = \sqrt {{{\left( {{x_l}\left[ n \right] - {x_k}} \right)}^2} + {{\left( {{y_l}\left[ n \right] - {y_k}} \right)}^2} + {{\left( {{z_l}\left[ n \right] - {z_k}} \right)}^2}}.
 \end{aligned}
 \end{equation}

 The $L$ UAVs serve $K$ UEs in a cooperative manner, and a flexible array architecture can be enabled by dynamically altering the positions of UAVs. When the movable region is large, the ground UEs are less likely to locate in the far-field region of the UAV swarm enabled AMA system. Thus, the more general near-field NUSW model needs to be considered. Denote by ${\beta _0} \triangleq {\left( {\frac{\lambda }{{4\pi }}} \right)^2}$ the channel power gain at the reference distance of ${r_0} = 1$ m, where $\lambda$ denotes the signal wavelength.
 For the LoS-dominating UAV-ground link, the near-field channel from UE $k$ to the UAV swarm enabled AMA system at slot $n$ is \cite{lu2024tutorial}
 \begin{equation}\label{nearFieldChannelUEk}
 \begin{aligned}
 &{{\bf{h}}_k}\left( {\left\{ {{{\bf{q}}_l}\left[ n \right]} \right\}} \right) = {\alpha _k}{{\bf{a}}_k}\left( {\left\{ {{{\bf{q}}_l}\left[ n \right]} \right\}} \right) = {\alpha _k}\times\\
 &{\left[ {\frac{{{r_k}}}{{{r_{k,1}}\left[ n \right]}}{e^{ - {\rm{j}}\frac{{2\pi }}{\lambda }\left( {{r_{k,1}}\left[ n \right] - {r_k}} \right)}}, \cdots ,\frac{{{r_k}}}{{{r_{k,L}}\left[ n \right]}}{e^{ - {\rm{j}}\frac{{2\pi }}{\lambda }\left( {{r_{k,L}}\left[ n \right] - {r_k}} \right)}}} \right]^T},
 \end{aligned}
 \end{equation}
 where ${\alpha _k} \triangleq \frac{{\sqrt {{\beta _0}} }}{{{r_k}}}{e^{ - {\rm{j}}\frac{{2\pi }}{\lambda }{r_k}}}$ denotes the complex-valued channel coefficient for the reference link distance $r_k$, and ${{\bf{a}}_k}\left( {\left\{ {{{\bf{q}}_l}\left[ n \right]} \right\}} \right) \in {\mathbb C}^{L \times 1}$ denotes the near-field NUSW array response vector. It is observed that the positions of the UAV swarm impact not only the channel amplitude, but also the channel phase. For simplicity, ${{\bf{h}}_k}\left[ n \right]$ and ${{\bf{a}}_k}\left[ n \right]$ are used to denote ${{\bf{h}}_k}\left( {\left\{ {{{\bf{q}}_l}\left[ n \right]} \right\}} \right)$ and ${{\bf{a}}_k}\left( {\left\{ {{{\bf{q}}_l}\left[ n \right]} \right\}} \right)$, respectively.

 Let ${s_k}\left[ {n,u} \right]$ represent the $u$-th independent and identically distributed (i.i.d.) information-bearing symbol of UE $k$ at slot $n$, with ${s_k}\left[ {n,u} \right] \sim {\cal CN}\left( {0,1} \right)$. By considering uplink communication and applying linear beamforming ${{\bf{v}}_k}\left[ n \right] \in {\mathbb C}^{L \times 1}$, with $\left\| {{{\bf{v}}_k}\left[ n \right]} \right\| = 1$, the resulting signal of UAV swarm enabled AMA system for UE $k$ at slot $n$ is
 \begin{equation}\label{receivedSignalUEk}
 \begin{aligned}
 &{y_k}\left[ {n,u} \right] = {\bf{v}}_k^H\left[ n \right]\sqrt {{P_k}} {{\bf{h}}_k}\left[ n \right]{s_k}\left[ {n,u} \right] + \\
 &\ \ \ \ \ {\bf{v}}_k^H\left[ n \right]\sum\limits_{i = 1,i \ne k}^K {\sqrt {{P_i}} {{\bf{h}}_i}\left[ n \right]{s_i}\left[ {n,u} \right]}  + {\bf{v}}_k^H\left[ n \right]{\bf{z}}\left[ {n,u} \right],
 \end{aligned}
 \end{equation}
 where ${P_i}$ represents the transmit power of UE $i$, and ${\bf{z}}\left[ {n,u} \right] \sim {\cal CN}( {{\bf{0}},{\sigma ^2}{{\bf{I}}_L}} )$ denotes the additive white Gaussian noise (AWGN) with zero mean and covariance matrix ${\sigma ^2}{{\bf{I}}_L}$. Let ${{\bar P}_i} \triangleq {P_i}/{\sigma ^2}$, $\forall i$, the resulting signal-to-interference-plus-noise ratio (SINR) for UE $k$ at slot $n$ can be expressed as
 \begin{equation}\label{SINRUEk}
 {\gamma _k}\left[ n \right] = \frac{{{{\bar P}_k}{{\left| {{\bf{v}}_k^H\left[ n \right]{{\bf{h}}_k}\left[ n \right]} \right|}^2}}}{{\sum\limits_{i = 1,i \ne k}^K {{{\bar P}_i}{{\left| {{\bf{v}}_k^H\left[ n \right]{{\bf{h}}_i}\left[ n \right]} \right|}^2}}  + 1}}.
 \end{equation}

 The achievable rate for UE $k$ at slot $n$ is
 \begin{equation}\label{achievableRateUEk}
 {R_k}\left[ n \right] = {\log _2}\left( {1 + {\gamma _k}\left[ n \right]} \right).
 \end{equation}
 Then, the average achievable rate for UE $k$ over the $N$ slots is given by
 \begin{equation}\label{averageAchievableRateUEk}
 {{\bar R}_k} = \frac{1}{N}\sum\limits_{n = 1}^N {{R_k}\left[ n \right]}.
 \end{equation}

 In this paper, we aim to jointly optimize the 3D UAV swarm formation/trajectory ${\bf{Q}} \triangleq \left\{ {{{\bf{q}}_l}\left[ n \right],\forall l,n} \right\}$ and receive beamforming for all UEs ${\bf{V}} \triangleq \left\{ {{{\bf{v}}_k}\left[ n \right],\forall k,n} \right\}$, to maximize the minimum average communication rate over the $K$ UEs. The optimization problem can be formulated as (by dropping the constant term)
 \begin{align}\label{OptimizationProblem}
 \left( {\rm{P1}} \right)\ \ &\mathop {\max }\limits_{{\bf{Q}}, {\bf{V}}}\ \mathop {\min }\limits_{k \in {\cal K}} \ \sum\limits_{n = 1}^N {{R_k}\left[ n \right]}\\
 {\rm{s.t.}}&\ \left\| {{{\bf{q}}_l}\left[ {n + 1} \right] - {{\bf{q}}_l}\left[ n \right]} \right\| \le {{\tilde V}_{\max }},\notag \\
 &\ \ \ \ \ \ \ \ \ \ \ \ \ \ \ \ \ \ \ \ \forall l\ {\rm and}\  n = 1, \cdots, N-1, \tag{\ref{OptimizationProblem}a}\\
 &\ {\left\| {{{\bf{q}}_l}\left[ n \right] - {{\bf{q}}_{l'}}\left[ n \right]} \right\|^2} \ge d_{\min }^2,\ \forall l, l' \ne l, n,\tag{\ref{OptimizationProblem}b}\\
 &\ {{\bf{q}}_l}\left[ 1 \right] = {{\bf{q}}_{I,l}},\ {{\bf{q}}_l}\left[ N \right] = {{\bf{q}}_{F,l}},\ \forall l,\tag{\ref{OptimizationProblem}c}\\
 &\ {z_l}\left[ n \right] \ge H,\ \forall l,n, \tag{\ref{OptimizationProblem}d}\\
 &\ \left\| {{{\bf{v}}_k}\left[ n \right]} \right\| = 1,\ \forall k,n, \tag{\ref{OptimizationProblem}e}
 \end{align}
 where ${{\bf{q}}_{I,l}}$ and ${{\bf{q}}_{F,l}}$ denote the initial and final positions of UAV $l$, respectively, and $H$ denotes the minimum altitude required to establish the LoS aerial-ground link. Note that problem (P1) is difficult to be solved directly due to the following two reasons. First, the optimization variables of UAV swarm trajectory and receive beamforming are coupled in the objective function, and the constraint (\ref{OptimizationProblem}b) is non-convex. Second, UAV swarm trajectory impacts both the channel amplitude and phase in complicated manners, rendering it challenging to be optimized.


\section{Special Cases of Single UE and TWO UEs}\label{sectionSingleTWOUE}
 In this section, we study the special cases of one single UE and two UEs, respectively. For ease of exposition, the constraints of initial and final positions are temporarily removed.

\subsection{Single UE}
 In this subsection, we consider the case of single UE. For ease of exposition, we assume that the location of UE is ${\bf{w}} = {\left[ {0,0,0} \right]^T}$, where the UE index is omitted for brevity.

 For the case of single UE, the IUI in \eqref{SINRUEk} is absent. With the optimal maximal-ratio combining (MRC) beamforming, the resulting SNR at slot $n$ is given by
 \begin{equation}\label{resultingSNR}
 \gamma \left[ n \right] = \bar P{\left\| {{\bf{h}}\left[ n \right]} \right\|^2} = \sum\limits_{l = 1}^L {\frac{{\bar P{\beta _0}}}{{r_l^2\left[ n \right]}}}  = \sum\limits_{l = 1}^L {\frac{{\bar P{\beta _0}}}{{{{\left\| {{{\bf{q}}_l}\left[ n \right] - {\bf{w}}} \right\|}^2}}}}.
 \end{equation}
 Then, problem (P1) is reduced to
 \begin{align}\label{OptimizationProblemSingleUESCA}
 \mathop {\max }\limits_{{\bf{Q}} } &\ \sum\limits_{n = 1}^N {{{\log }_2}\left( {1 + \sum\limits_{l = 1}^L {\frac{{\bar P{\beta _0}}}{{{{\left\| {{{\bf{q}}_l}\left[ n \right] - {\bf{w}}} \right\|}^2}}}} } \right)}\\
 {\rm{s.t.}}&\ \left\| {{{\bf{q}}_l}\left[ {n + 1} \right] - {{\bf{q}}_l}\left[ n \right]} \right\| \le {{\tilde V}_{\max }},\ \forall l\ {\rm and}\  n = 1, \cdots, N-1, \tag{\ref{OptimizationProblemSingleUESCA}a}\\
 &\ {\left\| {{{\bf{q}}_l}\left[ n \right] - {{\bf{q}}_{l'}}\left[ n \right]} \right\|^2} \ge d_{\min }^2,\ \forall l, l' \ne l, n.\tag{\ref{OptimizationProblemSingleUESCA}b}\\
 &\ {z_l}\left[ n \right] \ge H,\ \forall l,n.\tag{\ref{OptimizationProblemSingleUESCA}c}
 \end{align}

 A closer look at \eqref{OptimizationProblemSingleUESCA} shows that although the objective function is neither convex nor concave w.r.t. ${{{\bf{q}}_l}\left[ n \right]}$, it is convex w.r.t. ${{{\left\| {{{\bf{q}}_l}\left[ n \right] - {\bf{w}}} \right\|}^2}}$. Thus, the SCA technique \cite{zeng2019accessing} is applied to solve problem \eqref{OptimizationProblemSingleUESCA}. Specifically, by exploiting the property that a convex differentiable function is globally lower-bounded by its first-order Taylor approximation \cite{zeng2019accessing}, for given ${{{{\bf{\tilde q}}}_l}\left[ n \right]}$, we have
 \begin{equation}\label{collisionAvoidanceTaylor}
 \begin{aligned}
 R\left[ n \right] &\ge {R_{{\rm{lb}}}}\left[ n \right] \triangleq {\log _2}\left( {1 + \sum\limits_{l = 1}^L {\frac{{\bar P{\beta _0}}}{{{{\left\| {{{{\bf{\tilde q}}}_l}\left[ n \right] - {\bf{w}}} \right\|}^2}}}} } \right) - \\
 &\sum\limits_{l = 1}^L {\frac{{\frac{{\bar P{\beta _0}}}{{{{\left\| {{{{\bf{\tilde q}}}_l}\left[ n \right] - {\bf{w}}} \right\|}^4}}}}}{{\left( {1 + \sum\limits_{i = 1}^L {\frac{{\bar P{\beta _0}}}{{{{\left\| {{{{\bf{\tilde q}}}_i}\left[ n \right] - {\bf{w}}} \right\|}^2}}}} } \right)\ln 2}}}  \times \\
 &\ \ \ \ \ \ \ \ \left( {{{\left\| {{{\bf{q}}_l}\left[ n \right] - {\bf{w}}} \right\|}^2} - {{\left\| {{{{\bf{\tilde q}}}_l}\left[ n \right] - {\bf{w}}} \right\|}^2}} \right).
 \end{aligned}
 \end{equation}
 Moreover, the left-hand-side (LHS) of (\ref{OptimizationProblemSingleUESCA}b) is a convex function, we have the following inequality for given ${{{\bf{\tilde q}}}_l}\left[ n \right]$ and ${{{\bf{\tilde q}}}_{l'}}\left[ n \right]$
 \begin{equation}\label{interUAVDistancelb1}
 \begin{aligned}
 &{\left\| {{{\bf{q}}_l}\left[ n \right] - {{\bf{q}}_{l'}}\left[ n \right]} \right\|^2} \ge {\left[ {{{\left\| {{{\bf{q}}_l}\left[ n \right] - {{\bf{q}}_{l'}}\left[ n \right]} \right\|}^2}} \right]_{{\rm{lb}}}} \triangleq \\
 &\ \ \ 2{\left( {{{{\bf{\tilde q}}}_l}\left[ n \right] - {{{\bf{\tilde q}}}_{l'}}\left[ n \right]} \right)^T}\left( {{{\bf{q}}_l}\left[ n \right] - {{\bf{q}}_{l'}}\left[ n \right]} \right) - {\left\| {{{{\bf{\tilde q}}}_l}\left[ n \right] - {{{\bf{\tilde q}}}_{l'}}\left[ n \right]} \right\|^2}.
 \end{aligned}
 \end{equation}

 As a result, the objective function value of problem \eqref{OptimizationProblemSingleUESCA} is lower-bounded by the optimal value of the following problem
 \begin{equation}\label{OptimizationProblemSingleUESCAlb}
 \begin{aligned}
 \mathop {\max }\limits_{{\bf{Q}} } &\ \sum\limits_{n = 1}^N {R_{{\rm{lb}}}}\left[ n \right]\\
 {\rm{s.t.}}&\ \left\| {{{\bf{q}}_l}\left[ {n + 1} \right] - {{\bf{q}}_l}\left[ n \right]} \right\| \le {{\tilde V}_{\max }}, \forall l\ {\rm and}\  n = 1, \cdots, N-1,\\
 &\ {\left[ {{{\left\| {{{\bf{q}}_l}\left[ n \right] - {{\bf{q}}_{l'}}\left[ n \right]} \right\|}^2}} \right]_{{\rm{lb}}}} \ge d_{\min }^2,\ \forall l, l' \ne l, n,\\
 &\ {z_l}\left[ n \right] \ge H,\ \forall l,n.
 \end{aligned}
 \end{equation}
 Problem \eqref{OptimizationProblemSingleUESCAlb} is a convex optimization problem, and thus the standard convex optimization tools, e.g., CVX, can be applied to solve it. The complexity for solving \eqref{OptimizationProblemSingleUESCAlb} is ${\cal O}(I{\left( {NL} \right)^3})$, where $I$ represents the number of iterations required for convergence.


 It is worth noting that when $N = 1$ or $N$ is sufficiently large such that the  intermediate movement time slots can be ignored, UAV swarm trajectory optimization reduces to the placement optimization. In the following, the simple scenarios with $L = 1$ and $L =2$ are considered.

 \subsubsection{$L = 1$}
 The optimal single UAV placement position is ${{\bf{q}}^ {\star}} ={\left[ {0,0,H} \right]^T}$, i.e., placed directly above the UE to minimize the path loss.

 \subsubsection{$L = 2$}
 For the scenario with $L = 2$, when ${z_1} = {z_2} = H$, the placement optimization problem becomes
 \begin{equation}\label{optimizationProblemL2}
 \begin{aligned}
 \mathop {\max }\limits_{{{\bf{q}}_1},{{\bf{q}}_2}}&\  \bar P{\beta _0}\left( {\frac{1}{{{{\left\| {{{\bf{q}}_1} - {\bf{w}}} \right\|}^2}}} + \frac{1}{{{{\left\| {{{\bf{q}}_2} - {\bf{w}}} \right\|}^2}}}} \right)\\
 {\rm{s.t.}} &\ {\left\| {{{\bf{q}}_1} - {{\bf{q}}_{2}}} \right\|^2} \ge d_{\min }^2.
 \end{aligned}
 \end{equation}

 \begin{proposition} \label{UAVPlacementL2Proposition}
 An optimal solution to the UAV placement for problem \eqref{optimizationProblemL2} is ${\bf q}_1^{\star} = {\left[ {x^{\star},0,H } \right]^T}$ and ${\bf q}_2^{\star} = {\left[ { x^{\star} + d_{\min},0,H } \right]^T}$, where
 \begin{equation}\label{optimalPlacementL2}
 {x^ \star } = \left\{ \begin{split}
 &- \frac{{{d_{\min }}}}{2},\ {\rm if}\ \zeta  > \sqrt 3 /2,\\
 &- \frac{{{d_{\min }}}}{2} \pm {d_{\min }}\sqrt {\sqrt {{\zeta ^2} + \frac{1}{4}}  - \left( {{\zeta ^2} + \frac{1}{4}} \right)},\ {\rm otherwise},
 \end{split} \right.
 \end{equation}
 with $\zeta \triangleq H/{d_{\min }}$.
 \end{proposition}
 \begin{IEEEproof}
 Please refer to Appendix~\ref{proofofUAVPlacementL2Proposition}.
 \end{IEEEproof}

 Proposition~\ref{UAVPlacementL2Proposition} shows that when $\zeta  > \sqrt 3 /2$, the two UAVs should be placed symmetrically w.r.t. the UE's horizontal position. On the other hand, when $\zeta  \le \sqrt 3 /2$, the asymmetric placement is preferable, so as to improve the SNR.

 For arbitrary number of UAVs, the optimized UAV swarm placement positions can be obtained by solving the following problem:
 \begin{equation}\label{OptimizationProblemSingleUEJoint}
 \begin{aligned}
 \mathop {\max }\limits_{{\bf{Q}} } &\ {\log _2}\left( {1 + \sum\limits_{l = 1}^L {\frac{{\bar P{\beta _0}}}{{{{\left\| {{{\bf{q}}_l} - {\bf{w}}} \right\|}^2}}}} } \right)\\
 {\rm{s.t.}}&\ {\left\| {{{\bf{q}}_l} - {{\bf{q}}_{l'}}} \right\|^2} \ge d_{\min }^2,\ \forall l, l' \ne l,\\
 &\ {z_l} \ge H,\ \forall l,
 \end{aligned}
 \end{equation}
 which can be tackled by the SCA technique similar to \eqref{OptimizationProblemSingleUESCA}. The details are omitted for brevity.

 Furthermore, to further reduce the complexity, we propose a low-complexity successive placement algorithm. Specifically, the first UAV is directly placed above the UE, i.e., ${{\bf{q}}_1} = {\left[ {0,0,H} \right]^T}$. Then, one UAV is successively placed in each step to maximize the resulting SNR, while satisfying the collision avoidance constraint. For UAV $l$, $l \ge 2$, the placement optimization problem is
 \begin{equation}\label{OptimizationProblemSingleUESuccessive}
 \begin{aligned}
 \mathop {\max }\limits_{{{\bf{q}}_l}}&\ {\log _2}\left( {{C_l} + \frac{{{\bar P}{\beta _0}}}{{{{\left\| {{{\bf{q}}_l} - {\bf{w}}} \right\|}^2}}}} \right)\\
 {\rm{s.t.}}&\ {\left\| {{{\bf{q}}_l} - {{\bf{q}}_{l'}}} \right\|^2} \ge d_{\min }^2,\ 1 \le l' \le l - 1,\\
 &\ {z_l} \ge H,\ \forall l,
 \end{aligned}
 \end{equation}
 where ${C_l} \triangleq 1 + \sum\nolimits_{l' = 1}^{l - 1} {\frac{{{\bar P}{\beta _0}}}{{{{\left\| {{{\bf{q}}_{l'}} - {\bf{w}}} \right\|}^2}}}} $ is a constant term, and problem \eqref{OptimizationProblemSingleUESuccessive} can be solved with the SCA technique. Note that compared to joint  placement optimization in \eqref{OptimizationProblemSingleUEJoint} that has a complexity of ${\cal O}\left( {{I_1}{L^3}} \right)$, the complexity for successive placement algorithm is reduced to ${\cal O}\left( {{I_2}\left( {L - 1} \right)} \right)$, where $I_1$ and $I_2$ represent the number of iterations required for convergence, respectively.

\subsection{Two UEs}\label{subSectionTwoUEs}
 In this subsection, we study the placement optimization for the case of two UEs. Without loss of generality, the two UEs are located at ${{\bf{w}}_1} = {\left[ {X,0,0} \right]^T}$ and ${{\bf{w}}_2} = {\left[ { - X,0,0} \right]^T}$, respectively.

 It is worth noting that the SINR expression for UE $k$ in \eqref{SINRUEk} constitutes a generalized Rayleigh quotient in terms of the receive beamforming, and the optimal solution is given by the minimum mean-square error (MMSE) beamforming. This yields the following SINR expression \cite{lu2022near}
 \begin{equation}\label{SINRTwoUE}
 {\gamma _k} = {\bar P_k}{\left\| {{{\bf{h}}_k}} \right\|^2}\left( {1 - \frac{{{{\bar P}_{k'}}{{\left\| {{{\bf{h}}_{k'}}} \right\|}^2}{\rho _{k,k'}}}}{{{{\bar P}_{k'}}{{\left\| {{{\bf{h}}_{k'}}} \right\|}^2} + 1}}} \right),
 \end{equation}
 where $k,k' = 1,2$, $ k \ne k'$, and ${\rho _{k,k'}} \triangleq \frac{{{{\left| {{\bf{h}}_k^H{{\bf{h}}_{k'}}} \right|}^2}}}{{{{\left\| {{{\bf{h}}_k}} \right\|}^2}{{\left\| {{{\bf{h}}_{k'}}} \right\|}^2}}}$ represents the squared-correlation coefficient between the near-field channels of UE $k$ and $k'$. It is shown that the resulting SINR of UE $k$ at slot $n$ is impacted by ${\rho _{k,k'}}$, and thus we study this term in the following. Specifically, by substituting \eqref{nearFieldChannelUEk} into ${\rho _{k,k'}}$, we have
 \begin{equation}\label{correlationCoefficient}
 {\rho _{k,k'}} = \frac{{{{\left| {{\bf{a}}_k^H{{\bf{a}}_{k'}}} \right|}^2}}}{{{{\left\| {{{\bf{a}}_k}} \right\|}^2}{{\left\| {{{\bf{a}}_{k'}}} \right\|}^2}}} = \frac{{{{\left| {\sum\limits_{l = 1}^L {\frac{{{r_k}{r_{k'}}}}{{{r_{k,l}}{r_{k',l}}}}{e^{{\rm{j}}\frac{{2\pi }}{\lambda }\left( {{r_{k,l}} - {r_{k',l}}} \right)}}} } \right|}^2}}}{{{{\left\| {{{\bf{a}}_k}} \right\|}^2}{{\left\| {{{\bf{a}}_{k'}}} \right\|}^2}}}.
 \end{equation}

 \begin{proposition} \label{zerocorrelationCoefficientProposition}
 When $L$ is an even number, a UAV swarm trajectory achieving ${\rho _{k,k'}} = 0$ can be expressed in terms of ${z_l}$ as
 \begin{equation}\label{IUIFreeUAVTrajectoryTwo}
 {{\bf q}_l} = \left\{ \begin{array}{l}
 {\left[ {\frac{{\left( {2\nu  + 1} \right)\lambda }}{8}\sqrt {\frac{{z_l^2}}{{{X^2} - {{\left( {\frac{{\left( {2\nu  + 1} \right)\lambda }}{8}} \right)}^2}}} + 1} ,0,{z_l}} \right]^T},\\ \ \ \ \ \ \ \ \ \ \ \ \ \ \ \ \ \ \ \ \ \ \ \ \nu  \in {\mathbb Z}, 1 \le l \le \frac{L}{2},\\
 {\left[ { - \frac{{\left( {2\nu  + 1} \right)\lambda }}{8}\sqrt {\frac{{z_{l - \frac{L}{2}}^2}}{{{X^2} - {{\left( {\frac{{\left( {2\nu  + 1} \right)\lambda }}{8}} \right)}^2}}} + 1} ,0,{z_{l - \frac{L}{2}}}} \right]^T},\\ \ \ \ \ \ \ \ \ \ \ \ \ \ \ \ \ \ \ \ \ \ \ \ \nu  \in {\mathbb Z}, \frac{L}{2} + 1 \le l \le L.
 \end{array} \right.
 \end{equation}

 \end{proposition}
 \begin{IEEEproof}
 Please refer to Appendix~\ref{proofofzerocorrelationCoefficientProposition}.
 \end{IEEEproof}

 Proposition~\ref{zerocorrelationCoefficientProposition} shows that an IUI-free communication, i.e.,  ${\rho _{k,k'}} = 0$, can be achieved by following the hyperbolic placement given in \eqref{hyperbolicEquation} in Appendix~\ref{proofofzerocorrelationCoefficientProposition}, as illustrated in Fig.~\ref{fig:twoUEhyperbolic}. Since the $L$ UAVs are symmetrically arranged about the $z$-axis, the collision avoidance constraints among UAVs with $1 \le l \le L/2$ can be satisfied by adjusting the $z$-axis coordinate ${z_l}$, and the constraints for the remaining UAVs with $L/2 + 1 \le l \le L$ can be naturally satisfied. Moreover, UAVs on the left and right branches should satisfy the collision avoidance constraints, i.e., the distance between the first left and right UAVs with altitude $H$ should be no smaller than $d_{\min}$, i.e.,
 \begin{equation}\label{nuParameter}
 2\frac{{\left( {2\nu  + 1} \right)\lambda }}{8}\sqrt {\frac{{{H^2}}}{{{X^2} - {{\left( {\frac{{\left( {2\nu  + 1} \right)\lambda }}{8}} \right)}^2}}} + 1}  \ge {d_{\min }}.
 \end{equation}
 A closer look at the LHS of \eqref{nuParameter} reveals that it is a monotonic increasing function of $\nu$, and thus the minimum $\nu$ satisfying \eqref{nuParameter} can be efficiently found via the bisection method.


 \begin{figure}[!t]
 \centering
 \centerline{\includegraphics[width=3.5in,height=2.625in]{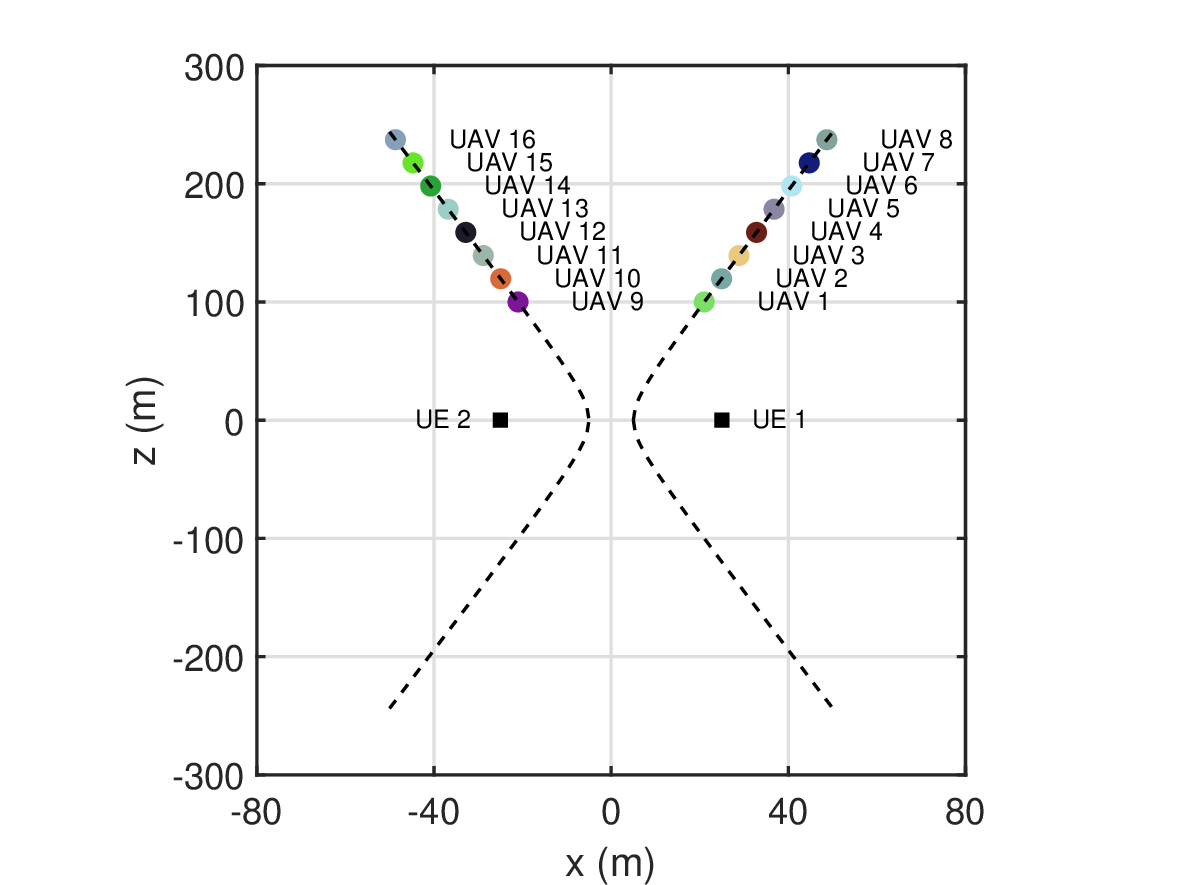}}
 \caption{An illustration of UAV swarm placement positions for IUI-free communication with two UEs.}
 \label{fig:twoUEhyperbolic}
 \end{figure}

 Moreover, with the UAV swarm placement position given in \eqref{IUIFreeUAVTrajectoryTwo}, the achievable rate of UE $k$, $k =1,2$, is
 \begin{equation}
 \begin{aligned}
 {R_k} = {\log _2}\left( {1 + \sum\limits_{l = 1}^{L/2} {\left( {\frac{{{{\bar P}_k}{\beta _0}}}{{{{\left( {{x_l} - X} \right)}^2} + z_l^2}} + \frac{{{{\bar P}_k}{\beta _0}}}{{{{\left( {{x_l} + X} \right)}^2} + z_l^2}}} \right)} } \right).
 \end{aligned}
 \end{equation}
 Then, problem (P1) is reduced to
 \begin{equation}
 \begin{aligned}
 \mathop {\max }\limits_{\left\{ {{z_l}} \right\}_{l = 1}^{L/2}} &\ \min \left( {{R_1},{R_2}} \right)\\
 {\rm{s.t.}}&\ {\left\| {{{\bf{q}}_l} - {{\bf{q}}_{l'}}} \right\|^2} \ge d_{\min }^2,\ \forall l, l' \ne l, \\
 &\ {z_l} \ge H,\ \forall l. \\
 \end{aligned}
 \end{equation}

 To tackle this problem, we propose a successive placement scheme. Specifically, let the $z$-coordinate value of the first UAV be $H$, i.e., ${z_1} = H$, and ${x_1}$ can be determined based on \eqref{IUIFreeUAVTrajectoryTwo}. Moreover, by leveraging the property that the distances between any position of the upper branch of the hyperbola and two foci increase with its $z$-coordinate value, the second UAV is placed on the right branch of the hyperbola such that $\left\| {{{\bf{q}}_2} - {{\bf{q}}_1}} \right\| = {d_{\min }}$, thus obtaining a larger SNR. The remaining UAVs' placement positions can be successively determined based on $\left\| {{{\bf{q}}_l} - {{\bf{q}}_{l - 1}}} \right\| = {d_{\min }}$, $\forall l \in \left[ {3,L/2} \right]$. 

%
%
%
%

\section{Arbitrary Number of UEs}\label{sectionArbitraryUE}
 In this section, we investigate the case with arbitrary number of UEs, and propose an efficient alternating optimization algorithm to solve (P1).

 It is worth noting that UAV swarm trajectory has an impact on both the channel amplitude and phase. By exploiting the fact that the channel phase is more sensitive to trajectory variation than the channel amplitude, we propose to separately consider the impact of trajectory variation on amplitude and phase. Specifically, we first obtain an initial UAV swarm trajectory by focusing on the channel amplitude to maximize the minimum average SNR over $N$ time slots across $K$ UEs (by temporarily neglecting the IUI). Then, by fixing the channel amplitude with the above obtained result, the receive beamforming and UAV swarm trajectory are optimized in an alternating manner.

 \subsection{UAV Swarm Trajectory Initialization}\label{subSectionUAVTrajectoryInitialization}
 In this subsection, we focus on the impact of UAV swarm trajectory on the channel amplitude. By applying the MRC beamforming ${{\bf{v}}_k}\left[ n \right] = {{\bf{h}}_k}\left[ n \right]/\left\| {{{\bf{h}}_k}\left[ n \right]} \right\|$ to \eqref{SINRUEk}, the average SNR without considering the IUI for UE $k$ over $N$ time slots is
 \begin{equation}\label{averageSNRWithoutIUI}
 {{\bar \gamma }_k} = \frac{1}{N}\sum\limits_{n = 1}^N {{{\bar P}_k}{{\left\| {{{\bf{h}}_k}\left[ n \right]} \right\|}^2}}  = \frac{1}{N}\sum\limits_{n = 1}^N {\sum\limits_{l = 1}^L {\frac{{{{\bar P}_k}{\beta _0}}}{{{{\left\| {{{\bf{q}}_l}\left[ n \right] - {{\bf{w}}_k}} \right\|}^2}}}} }, \ \forall k.
 \end{equation}

 We aim to maximize the minimum average SNR in \eqref{averageSNRWithoutIUI} over $K$ UEs, by optimizing the UAV swarm trajectory. The optimization problem is formulated as (by dropping the constant term)
 \begin{equation}\label{trajectoryInitializationProblem}
 \begin{aligned}
 \mathop {\max }\limits_{{\bf{Q}}}&\ \mathop {\min }\limits_{k \in {\cal K}} \ \sum\limits_{n = 1}^N {\sum\limits_{l = 1}^L {\frac{{{{\bar P}_k}{\beta _0}}}{{{{\left\| {{{\bf{q}}_l}\left[ n \right] - {{\bf{w}}_k}} \right\|}^2}}}} } \\
 {\rm{s.t.}}&\ \left\| {{{\bf{q}}_l}\left[ {n + 1} \right] - {{\bf{q}}_l}\left[ n \right]} \right\| \le {{\tilde V}_{\max }},\ \forall l\ {\rm and}\  n = 1, \cdots, N-1,\\
 &\ {\left\| {{{\bf{q}}_l}\left[ n \right] - {{\bf{q}}_{l'}}\left[ n \right]} \right\|^2} \ge d_{\min }^2,\ \forall l, l' \ne l, n,\\
 &\ {{\bf{q}}_l}\left[ 1 \right] = {{\bf{q}}_{I,l}},\ {{\bf{q}}_l}\left[ N \right] = {{\bf{q}}_{F,l}},\ \forall l,\\
 &\ {z_l}\left[ n \right] \ge H,\ \forall l,n.
 \end{aligned}
 \end{equation}

 In the following, SCA technique is applied to tackle the non-convex optimization problem \eqref{trajectoryInitializationProblem}. We first introduce the slack variable $\Psi$ to transform problem \eqref{trajectoryInitializationProblem} into the following problem:
 \begin{align}\label{equiTrajectoryInitializationProblem}
 \mathop {\max }\limits_{{\bf{Q}},\Psi}&\ \Psi  \\
 {\rm{s.t.}}&\ \sum\limits_{n = 1}^N {\sum\limits_{l = 1}^L {\frac{{{{\bar P}_k}{\beta _0}}}{{{{\left\| {{{\bf{q}}_l}\left[ n \right] - {{\bf{w}}_k}} \right\|}^2}}}} }  \ge \Psi,\ \forall k, \tag{\ref{equiTrajectoryInitializationProblem}a}\\
 &\ \left\| {{{\bf{q}}_l}\left[ {n + 1} \right] - {{\bf{q}}_l}\left[ n \right]} \right\| \le {{\tilde V}_{\max }},\ \forall l\ {\rm and}\  n = 1, \cdots, N-1,\tag{\ref{equiTrajectoryInitializationProblem}b}\\
 &\ {\left\| {{{\bf{q}}_l}\left[ n \right] - {{\bf{q}}_{l'}}\left[ n \right]} \right\|^2} \ge d_{\min }^2,\ \forall l, l' \ne l, n,\tag{\ref{equiTrajectoryInitializationProblem}c}\\
 &\ {{\bf{q}}_l}\left[ 1 \right] = {{\bf{q}}_{I,l}},\ {{\bf{q}}_l}\left[ N \right] = {{\bf{q}}_{F,l}},\ \forall l,\tag{\ref{equiTrajectoryInitializationProblem}d}\\
  &\ {z_l}\left[ n \right] \ge H,\ \forall l,n.\tag{\ref{equiTrajectoryInitializationProblem}e}
 \end{align}

 A closer look at the LHS of (\ref{equiTrajectoryInitializationProblem}a) shows that $1/{\left\| {{{\bf{q}}_l}\left[ n \right] - {{\bf{w}}_k}} \right\|^2}$ is a convex function w.r.t. $\left\| {{{\bf{q}}_l}\left[ n \right] - {{\bf{w}}_k}} \right\|$. For given ${{\bf{\tilde q}}_l}\left[ n \right]$, it is lower-bounded by
 \begin{equation}
 \begin{aligned}
 &\frac{1}{{{{\left\| {{{\bf{q}}_l}\left[ n \right] - {{\bf{w}}_k}} \right\|}^2}}} \ge {\left[ {\frac{1}{{{{\left\| {{{\bf{q}}_l}\left[ n \right] - {{\bf{w}}_k}} \right\|}^2}}}} \right]_{{\rm{lb}}}} \triangleq\frac{1}{{{{\left\| {{{{\bf{\tilde q}}}_l}\left[ n \right] - {{\bf{w}}_k}} \right\|}^2}}} - \\
 &\ \ \ \ \ \frac{2}{{{{\left\| {{{{\bf{\tilde q}}}_l}\left[ n \right] - {{\bf{w}}_k}} \right\|}^3}}}\left( {\left\| {{{\bf{q}}_l}\left[ n \right] - {{\bf{w}}_k}} \right\| - \left\| {{{{\bf{\tilde q}}}_l}\left[ n \right] - {{\bf{w}}_k}} \right\|} \right).
 \end{aligned}
 \end{equation}

 As a result, given ${{\bf{\tilde q}}_l}\left[ n \right]$, we can solve the following optimization problem:
 \begin{equation}\label{equiTrajectoryInitializationProblem2}
 \begin{aligned}
 \mathop {\max }\limits_{{\bf{Q}},\Psi}&\ \Psi  \\
 {\rm{s.t.}}&\ \sum\limits_{n = 1}^N {\sum\limits_{l = 1}^L {{{\bar P}_k}{\beta _0}{{\left[ {\frac{1}{{{{\left\| {{{\bf{q}}_l}\left[ n \right] - {{\bf{w}}_k}} \right\|}^2}}}} \right]}_{{\rm{lb}}}}} }  \ge \Psi ,\ \forall k,\\
 &\ \left\| {{{\bf{q}}_l}\left[ {n + 1} \right] - {{\bf{q}}_l}\left[ n \right]} \right\| \le {{\tilde V}_{\max }},\ \forall l\ {\rm and}\  n = 1, \cdots, N-1,\\
 &\ {\left[ {{{\left\| {{{\bf{q}}_l}\left[ n \right] - {{\bf{q}}_{l'}}\left[ n \right]} \right\|}^2}} \right]_{{\rm{lb}}}} \ge d_{\min }^2,\ \forall l, l' \ne l, n,\\
 &\ {{\bf{q}}_l}\left[ 1 \right] = {{\bf{q}}_{I,l}},\ {{\bf{q}}_l}\left[ N \right] = {{\bf{q}}_{F,l}},\ \forall l,\\
 &\ {z_l}\left[ n \right] \ge H,\ \forall l,n,
 \end{aligned}
 \end{equation}
 which is a convex optimization problem, and thus can be efficiently solved via CVX.

 The obtained solution to \eqref{trajectoryInitializationProblem}, denoted as $\left\{ {{{{\bf{\hat q}}}_l}\left[ n \right]} \right\}$, serves as the initial UAV swarm trajectory for the subsequent alternating optimization. In particular, the channel amplitude is fixed with $\left\{ {{{{\bf{\hat q}}}_l}\left[ n \right]} \right\}$, and the impact of UAV swarm trajectory on the channel phase is considered in the alternating optimization. Such an approximation is valid since the wavelength-scale adjustment yields a significant variation in terms of the channel phase, while having a negligible impact on the channel amplitude. Then, for alternating optimization, the channel from UE $k$ to the UAV swarm enabled AMA system at slot $n$ is approximated as
 \begin{equation}\label{approxNearFieldChannelUEk}
 \begin{aligned}
 {{\bf{h}}_k}\left[ n \right] &= {\alpha _k}\left[ {\frac{{{r_k}}}{{{{\hat r}_{k,1}}\left[ n \right]}}{e^{ - {\rm{j}}\frac{{2\pi }}{\lambda }\left( {{r_{k,1}}\left[ n \right] - {r_k}} \right)}}, \cdots ,} \right.\\
 &\ \ \ \ \ \ \ \ \ \ \ \ \ \ \ \ {\left. {\frac{{{r_k}}}{{{{\hat r}_{k,L}}\left[ n \right]}}{e^{ - {\rm{j}}\frac{{2\pi }}{\lambda }\left( {{r_{k,L}}\left[ n \right] - {r_k}} \right)}}} \right]^T},
 \end{aligned}
 \end{equation}
 where ${{\hat r}_{k,l}}\left[ n \right] \triangleq \left\| {{{{\bf{\hat q}}}_l}\left[ n \right] - {{\bf{w}}_k}} \right\|$, $\forall l,k,n$.

 \subsection{Receive Beamforming and UAV Swarm Trajectory Optimization}
 For given UAV swarm trajectory ${\bf Q}$, the optimal receive beamforming is given by the MMSE solution \cite{lu2025UAVMA}
 \begin{equation}\label{optimalReceiveBeamforming}
 {{\bf{v}}_k}\left[ n \right] = \frac{{{\bf{C}}_k^{ - 1}\left[ n \right]{{\bf{h}}_k}\left[ n \right]}}{{\left\| {{\bf{C}}_k^{ - 1}\left[ n \right]{{\bf{h}}_k}\left[ n \right]} \right\|}},\ \forall k,n,
 \end{equation}
 where ${{\bf{C}}_k}\left[ n \right] \triangleq {{\bf{I}}_L} + \sum\nolimits_{i = 1,i \ne k}^K {{{\bar P}_i}{{\bf{h}}_i}\left[ n \right]{\bf{h}}_i^H\left[ n \right]} $ represents the interference-plus-noise covariance matrix of UE $k$ at slot $n$.

 On the other hand, for given receive beamforming ${\bf V}$, the sub-problem of (P1) for optimizing UAV swarm trajectory is
 \begin{equation}\label{originalTrajectoryOptimizationProblem}
 \begin{aligned}
 \mathop {\max }\limits_{{\bf{Q}}}&\ \mathop {\min }\limits_{k \in {\cal K}} \ \sum\limits_{n = 1}^N {{R_k}\left[ n \right]}\\
 {\rm{s.t.}}&\ \left\| {{{\bf{q}}_l}\left[ {n + 1} \right] - {{\bf{q}}_l}\left[ n \right]} \right\| \le {{\tilde V}_{\max }},\ \forall l\ {\rm and}\  n = 1, \cdots, N-1,\\
 &\ {\left\| {{{\bf{q}}_l}\left[ n \right] - {{\bf{q}}_{l'}}\left[ n \right]} \right\|^2} \ge d_{\min }^2,\ \forall l, l' \ne l, n,\\
 &\ {{\bf{q}}_l}\left[ 1 \right] = {{\bf{q}}_{I,l}},\ {{\bf{q}}_l}\left[ N \right] = {{\bf{q}}_{F,l}},\ \forall l,\\
  &\ {z_l}\left[ n \right] \ge H,\ \forall l,n.
 \end{aligned}
 \end{equation}
 In the following, each UAV trajectory is optimized alternatively in an iterative manner.



 By introducing the slack variables $\left\{ {{\eta _k}\left[ n \right],{\mu _k}\left[ n \right],\forall k,n} \right\}$ and $\Gamma $, the sub-problem of \eqref{originalTrajectoryOptimizationProblem} for optimizing UAV $l$'s trajectory with given $\left\{ {{{\bf{q}}_{l'}}\left[ n \right],\forall l' \ne l} \right\}_{n = 1}^N$ can be transformed into
 \begin{align}\label{trajectoryOptimizationProblem}
 &\ \mathop {\max }\limits_{\left\{ {{{\bf{q}}_l}\left[ n \right]} \right\}_{n = 1}^N,\left\{ {{\eta _k}\left[ n \right],{\mu _k}\left[ n \right]} \right\},\Gamma }\ \Gamma \\
 {\rm{s.t.}}&\ \sum\limits_{n = 1}^N {\ln \left( {{e^{\left( {{\eta _k}\left[ n \right] - {\mu _k}\left[ n \right]} \right)}}} \right)}  \ge \Gamma \ln 2,\ \forall k,\tag{\ref{trajectoryOptimizationProblem}a}\\
 & 1 + \sum\limits_{i = 1}^K {{{\bar P}_i}{{\left| {{\bf{v}}_k^H\left[ n \right]{{\bf{h}}_i}\left[ n \right]} \right|}^2}}  \ge {e^{{\eta _k}\left[ n \right]}},\ \forall k,n,\tag{\ref{trajectoryOptimizationProblem}b}\\
 & 1 + \sum\limits_{i = 1,i \ne k}^K {{{\bar P}_i}{{\left| {{\bf{v}}_k^H\left[ n \right]{{\bf{h}}_i}\left[ n \right]} \right|}^2}}  \le {e^{{\mu _k}\left[ n \right]}},\ \forall k,n,\tag{\ref{trajectoryOptimizationProblem}c}\\
 &\ \left\| {{{\bf{q}}_l}\left[ {n + 1} \right] - {{\bf{q}}_l}\left[ n \right]} \right\| \le {{\tilde V}_{\max }},\ \forall n = 1, \cdots ,N - 1,\tag{\ref{trajectoryOptimizationProblem}d}\\
 &\ {\left\| {{{\bf{q}}_l}\left[ n \right] - {{\bf{q}}_{l'}}\left[ n \right]} \right\|^2} \ge {d_{\min}^2},\ \forall l' \ne l,n,\tag{\ref{trajectoryOptimizationProblem}e}\\
 &\ {{\bf{q}}_l}\left[ 1 \right] = {{\bf{q}}_{I,l}},\ {{\bf{q}}_l}\left[ N \right] = {{\bf{q}}_{F,l}},\tag{\ref{trajectoryOptimizationProblem}f}\\
 &\ {z_l}\left[ n \right] \ge H,\ \forall n,\tag{\ref{trajectoryOptimizationProblem}g}
 \end{align}
 which is a non-convex optimization problem due to the constraints (\ref{trajectoryOptimizationProblem}b), (\ref{trajectoryOptimizationProblem}c), and (\ref{trajectoryOptimizationProblem}e).


 Let ${{h_{i,l}}\left[ n \right]}$ and ${v_{k,l}}\left[ n \right]$ represent the $l$-th element of ${{{\bf{h}}_i}\left[ n \right]}$ and ${{\bf{v}}_k}\left[ n \right]$, respectively, ${{{\bf{\bar h}}_{i,l}}\left[ n \right]} \in {\mathbb C}^{\left(L -1 \right) \times 1}$ and ${{{\bf{\bar v}}}_{k,l}}\left[ n \right] \in {\mathbb C}^{\left(L -1 \right) \times 1}$ represent the resulting vectors after removing ${{h_{i,l}}\left[ n \right]}$ and ${v_{k,l}}\left[ n \right]$ from ${{{\bf{h}}_i}\left[ n \right]}$ and ${{\bf{v}}_k}\left[ n \right]$, respectively. After some manipulations, ${g_{k,i}}\left[ n \right] \triangleq {\left| {{\bf{v}}_k^H\left[ n \right]{{\bf{h}}_i}\left[ n \right]} \right|^2}$ can be expressed as
 \begin{equation}\label{gkil}
 \begin{aligned}
 &{g_{k,i}}\left[ n \right] = \sum\limits_{l' = 1,l' \ne l}^L {2{\rm{Re}}\left\{ {h_{i,l}^*\left[ n \right]{V_{k,l,l'}}\left[ n \right]{h_{i,l'}}\left[ n \right]} \right\}}\\
 &+ h_{i,l}^*\left[ n \right]{V_{k,l,l}}\left[ n \right]{h_{i,l}}\left[ n \right] + {\bf{\bar h}}_{i,l}^H\left[ n \right]{{{\bf{\bar V}}}_{k,l,l}}\left[ n \right]{{{\bf{\bar h}}}_{i,l}}\left[ n \right]\\
 &= \underbrace {\sum\limits_{l' = 1,l' \ne l}^L {2{\mathop{\rm Re}\nolimits} \left\{ {h_{i,l}^ * \left[ n \right]{V_{k,l,l'}}\left[ n \right]{h_{i,l'}}\left[ n \right]} \right\}} }_{{f_{k,i,l}}\left[ n \right]}\\
 & + \underbrace {\frac{{{\beta _0}}}{{\hat r_{i,l}^2\left[ n \right]}}{V_{k,l,l}}\left[ n \right] + {\bf{\bar h}}_{i,l}^H\left[ n \right]{{{\bf{\bar V}}}_{k,l,l}}\left[ n \right]{{{\bf{\bar h}}}_{i,l}}\left[ n \right]}_{{{\bar f}_{k,i,l}}\left[ n \right]},\\
 \end{aligned}
 \end{equation}
 where ${V_{k,l,l'}}\left[ n \right] \triangleq {v_{k,l}}\left[ n \right]v_{k,l'}^*\left[ n \right]$ and ${{{\bf{\bar V}}}_{k,l,l}}\left[ n \right]  \triangleq {{{\bf{\bar v}}}_{k,l}}\left[ n \right]{\bf{\bar v}}_{k,l}^H\left[ n \right] \in {{\mathbb C}^{\left( {L - 1} \right) \times \left( {L - 1} \right)}}$.

 Furthermore, ${f_{k,i,l}}\left[ n \right]$ can be explicitly expressed in terms of ${{\bf{q}}_l}\left[ n \right]$ as
 \begin{equation}\label{fkil}
 \begin{aligned}
 &{f_{k,i,l}}\left[ n \right] = \sum\limits_{l' = 1,l' \ne l}^L {\frac{{2{{\left| {{\alpha _i}} \right|}^2}\left| {{V_{k,l,l'}}\left[ n \right]} \right|r_i^2}}{{{{\hat r}_{i,l}}\left[ n \right]{{\hat r}_{i,l'}}\left[ n \right]}} \times } \\
  &\ \ {\mathop{\rm Re}\nolimits} \left\{ {{e^{{\rm{j}}\left( {\frac{{2\pi }}{\lambda }\left( {{r_{i,l}}\left[ n \right] - {r_{i,l'}}\left[ n \right]} \right) + \angle {V_{k,l,l'}}\left[ n \right]} \right)}}} \right\}\\
 &= \sum\limits_{l' = 1,l' \ne l}^L {\frac{{2{{\left| {{\alpha _i}} \right|}^2}\left| {{V_{k,l,l'}}\left[ n \right]} \right|r_i^2}}{{{{\hat r}_{i,l}}\left[ n \right]{{\hat r}_{i,l'}}\left[ n \right]}}}  \times \\
 &\ \ \cos \left( {\frac{{2\pi }}{\lambda }\left( {\left\| {{{\bf{q}}_l}\left[ n \right] - {{\bf{w}}_i}} \right\| - \left\| {{{\bf{q}}_{l'}}\left[ n \right] - {{\bf{w}}_i}} \right\|} \right) + \angle {V_{k,l,l'}}\left[ n \right]} \right).
 \end{aligned}
 \end{equation}
 It is observed from \eqref{gkil} that the first term ${f_{k,i,l}}\left[ n \right]$ depends on ${{\bf{q}}_l}\left[ n \right]$, while the second term ${{\bar f}_{k,i,l}}\left[ n \right]$ is independent of it. Moreover, as can be seen in \eqref{fkil}, ${f_{k,i,l}}\left[ n \right]$ is neither convex nor concave w.r.t. ${{\bf{q}}_l}\left[ n \right]$. To tackle this issue, we construct two surrogate functions serving as the lower and upper bounds of ${f_{k,i,l}}\left[ n \right]$, respectively.

 \begin{lemma}\label{quadraticSurrogatelemma}
 For given ${{{\bf{\tilde q}}}_l}\left[ n \right]$, the lower and upper bounds of ${f_{k,i,l}}\left[ n \right]$ are given by
 \begin{equation}\label{lowerBoundfkim}
 \begin{aligned}
 &{f_{k,i,l}}\left[ n \right] \ge {\left[ {{f_{k,i,l}}\left[ n \right]} \right]_{{\rm{lb}}}} = {f_{k,i,l}}\left( {{{{\bf{\tilde q}}}_l}\left[ n \right]} \right) +  \\
 &\ \ \ \ \ \ \ \nabla {f_{k,i,l}}{\left( {{{{\bf{\tilde q}}}_l}\left[ n \right]} \right)^T}\left( {{{\bf{q}}_l}\left[ n \right] - {{{\bf{\tilde q}}}_l}\left[ n \right]} \right)\\
 &- \frac{{{\omega _{k,i,l}}\left[ n \right]}}{2}{\left( {{{\bf{q}}_l}\left[ n \right] - {{{\bf{\tilde q}}}_l}\left[ n \right]} \right)^T}\left( {{{\bf{q}}_l}\left[ n \right] - {{{\bf{\tilde q}}}_l}\left[ n \right]} \right),
 \end{aligned}
 \end{equation}
 \begin{equation}\label{upperBoundfkim}
 \begin{aligned}
 &{f_{k,i,l}}\left[ n \right] \le {\left[ {{f_{k,i,l}}\left[ n \right]} \right]_{{\rm{ub}}}} = {f_{k,i,l}}\left( {{{{\bf{\tilde q}}}_l}\left[ n \right]} \right) + \\
 &\ \ \ \ \ \ \ \nabla {f_{k,i,l}}{\left( {{{{\bf{\tilde q}}}_l}\left[ n \right]} \right)^T}\left( {{{\bf{q}}_l}\left[ n \right] - {{{\bf{\tilde q}}}_l}\left[ n \right]} \right)\\
 &+\frac{{{\omega_{k,i,l}}\left[ n \right]}}{2}{\left( {{{\bf{q}}_l}\left[ n \right] - {{{\bf{\tilde q}}}_l}\left[ n \right]} \right)^T}\left( {{{\bf{q}}_l}\left[ n \right] - {{{\bf{\tilde q}}}_l}\left[ n \right]} \right),
 \end{aligned}
 \end{equation}
 where ${\omega _{k,i,l}}\left[ n \right] \triangleq \frac{{4{\pi ^2}}}{{{\lambda ^2}}}\sum\nolimits_{l' = 1,l' \ne l}^L {{\zeta _{k,i,l,l'}}\left[ n \right]} $, with ${\zeta _{k,i,l,l'}}\left[ n \right] \triangleq \frac{{2{{\left| {{\alpha _i}} \right|}^2}\left| {{V_{k,l,l'}}\left[ n \right]} \right|r_i^2}}{{{{\hat r}_{i,l}}\left[ n \right]{{\hat r}_{i,l'}}\left[ n \right]}}$, and $\nabla {f_{k,i,l}}({{\bf{\tilde q}}_l}\left[ n \right])$ denotes the gradient of ${{f_{k,i,l}}\left[ n \right]}$ over ${{\bf{\tilde q}}_l}\left[ n \right]$, which can be obtained based on  \eqref{firstOrderDerivative} in Appendix \ref{proofOfquadraticSurrogatelemma}.
 \end{lemma}

 \begin{IEEEproof}
 Please refer to Appendix \ref{proofOfquadraticSurrogatelemma}.
 \end{IEEEproof}

 With Lemma~\ref{quadraticSurrogatelemma}, the lower and upper bounds of ${g_{k,i}}\left[ n \right]$ are given by ${\left[ {{g_{k,i}}\left[ n \right]} \right]_{{\rm{lb}}}} = {\left[ {{f_{k,i,l}}\left[ n \right]} \right]_{{\rm{lb}}}} + {{\bar f}_{k,i,l}}\left[ n \right]$ and ${\left[ {{g_{k,i}}\left[ n \right]} \right]_{{\rm{ub}}}} = {\left[ {{f_{k,i,l}}\left[ n \right]} \right]_{{\rm{ub}}}} + {{\bar f}_{k,i,l}}\left[ n \right]$.

 Furthermore, for given ${{\tilde \mu }_k}\left[ n \right]$, the right-hand-side (RHS) of (\ref{trajectoryOptimizationProblem}c) is lower-bounded by
 \begin{equation}
 {e^{{\mu _k}\left[ n \right]}} \ge {\left[ {{e^{{\mu _k}\left[ n \right]}}} \right]_{\rm lb}} \triangleq \left( {1 + {\mu _k}\left[ n \right] - {{\tilde \mu }_k}\left[ n \right]} \right){e^{{{\tilde \mu }_k}\left[ n \right]}}.
 \end{equation}
 The LHS of (\ref{trajectoryOptimizationProblem}e) is a convex function w.r.t. ${{{\bf{q}}_l}\left[ n \right]}$. For given ${{{\bf{\tilde q}}}_l}\left[ n \right]$, it is lower-bounded by
 \begin{equation}
 \begin{aligned}
 &{\left\| {{{\bf{q}}_l}\left[ n \right] - {{\bf{q}}_{l'}}\left[ n \right]} \right\|^2} \ge {\left[ {{{\left\| {{{\bf{q}}_l}\left[ n \right] - {{\bf{q}}_{l'}}\left[ n \right]} \right\|}^2}} \right]_{{\rm{lb}},2}} \triangleq \\
 &{\left\| {{{{\bf{\tilde q}}}_l}\left[ n \right] - {{\bf{q}}_{l'}}\left[ n \right]} \right\|^2} + 2{\left( {{{{\bf{\tilde q}}}_l}\left[ n \right] - {{\bf{q}}_{l'}}\left[ n \right]} \right)^T}\left( {{{\bf{q}}_l}\left[ n \right] - {{{\bf{\tilde q}}}_l}\left[ n \right]} \right).
\end{aligned}
 \end{equation}

 As a result, the objective function value of problem \eqref{trajectoryOptimizationProblem} is lower-bounded by that of the following problem for given ${{{\bf{\tilde q}}}_l}\left[ n \right]$ and ${{\tilde \mu }_k}\left[ n \right]$
 \begin{equation}\label{trajectoryOptimizationProblemlb}
 \begin{aligned}
 &\ \mathop {\max }\limits_{\left\{ {{{\bf{q}}_l}\left[ n \right]} \right\}_{n = 1}^N,\left\{ {{\eta _k}\left[ n \right],{\mu _k}\left[ n \right]} \right\},\Gamma } \Gamma \\
 {\rm{s.t.}}&\ \sum\limits_{n = 1}^N {\left( {{\eta _k}\left[ n \right] - {\mu _k}\left[ n \right]} \right)}  \ge \Gamma \ln 2,\ \forall k,\\
 & 1 + \sum\limits_{i = 1}^K {{{\bar P}_i}{{\left[ {{g_{k,i}}\left[ n \right]} \right]}_{{\rm{lb}}}}}  \ge {e^{{\eta _k}\left[ n \right]}},\ \forall k,n,\\
 & 1 + \sum\limits_{i = 1,i \ne k}^K {{{\bar P}_i}{{\left[ {{g_{k,i}}\left[ n \right]} \right]}_{{\rm{ub}}}}}  \le {\left[ {{e^{{\mu _k}\left[ n \right]}}} \right]_{{\rm{lb}}}},\ \forall k,n,\\
 &\ \left\| {{{\bf{q}}_l}\left[ {n + 1} \right] - {{\bf{q}}_l}\left[ n \right]} \right\| \le {{\tilde V}_{\max }},\ \forall n = 1, \cdots ,N - 1,\\
 &\ {\left[ {{{\left\| {{{\bf{q}}_l}\left[ n \right] - {{\bf{q}}_{l'}}\left[ n \right]} \right\|}^2}} \right]_{{\rm{lb}},2}} \ge {d_{\min}^2},\ \forall l' \ne l,n,\\
 &\ {{\bf{q}}_l}\left[ 1 \right] = {{\bf{q}}_{I,l}},\ {{\bf{q}}_l}\left[ N \right] = {{\bf{q}}_{F,l}},\\
 &\ {z_l}\left[ n \right] \ge H,\ \forall n.
 \end{aligned}
 \end{equation}
 It can be seen that problem \eqref{trajectoryOptimizationProblemlb} is a convex optimization problem, and thus can be solved efficiently via CVX.

 With the optimized UAV swarm trajectory $\left\{ {{\bf{q}}_l^{\star}\left[ n \right],\forall l,n} \right\}$, we substitute it into \eqref{nearFieldChannelUEk}, i.e., ${h_{k,l}}\left[ n \right] = $ $ {\alpha _k}\frac{{{r_k}}}{{\left\| {{\bf{q}}_l^ \star \left[ n \right] - {{\bf{w}}_k}} \right\|}}{e^{ - {\rm{j}}\frac{{2\pi }}{\lambda }\left( {\left\| {{\bf{q}}_l^ \star \left[ n \right] - {{\bf{w}}_k}} \right\| - {r_k}} \right)}}$, and an additional MMSE receive beamforming is performed as in \eqref{optimalReceiveBeamforming}.

 The proposed alternating optimization algorithm for solving (P1) is summarized in Algorithm~\ref{alg1}. 
 The computational complexity of Algorithm~\ref{alg1} is analyzed as follows. From step 2 to step 5, the computational complexity for obtaining the UAV swarm trajectory is ${\cal O}( {{I_3}{{(LN)}^3}} )$, where $I_3$ denotes the number of iterations required for convergence. In step 8, the complexity for calculating the receive beamforming is ${\cal O}\left( {NK{L^3}} \right)$. From step 9 to step 11, the complexity is approximately given by ${\cal O}( {{I_4}L{{\left( {NK} \right)}^3}} )$, where $I_4$ is the maximum number of iterations required for convergence. As a result, the total computational complexity of Algorithm~\ref{alg1} is approximately given by ${\cal O}( {I_3}{(LN)^3} + {I_5}NK{L^3} + {I_4}{I_5}L{\left( {NK} \right)^3} )$, where $I_5$ represents the number of iterations required by the alternating optimization algorithm to converge.

 \begin{algorithm}[t]
 \caption{Proposed Alternating Optimization Algorithm for Solving (P1)}
 \label{alg1}
 \begin{algorithmic}[1]
 \STATE Randomly initialize ${{\bf{Q}}^{\left( 0 \right)}}$, and let $j=0$.
 \REPEAT
 \STATE Obtain the UAV swarm trajectory by solving the convex problem \eqref{equiTrajectoryInitializationProblem2}, denoted as ${{\bf{Q}}^{\star}}$.
 \STATE Update $j = j + 1$, and ${{\bf{Q}}^{\left( {j} \right)}} = {{\bf{Q}}^{\star}}$.
 \UNTIL the fractional increase in the objective function value is below a given threshold $ \epsilon_1$.
 \STATE Denote the optimized trajectory as ${\bf{\hat Q}} = \left\{ {{{{\bf{\hat q}}}_l}\left[ n \right],\forall l,n} \right\}$, and initialize ${{\bf{Q}}^{\left( 0 \right)}} = {\bf{\hat Q}}$. Let $b=0$.
 \REPEAT
 \STATE For given ${{\bf{Q}}^{\left( b \right)}}$, obtain the optimal receive beamforming based on \eqref{optimalReceiveBeamforming}, denoted as ${{\bf{V}}^{\left( {b + 1} \right)}}$.
 \STATE \textbf{for} $l = 1:L$ \textbf{do}
 \STATE \quad Solve problem \eqref{trajectoryOptimizationProblemlb} for given ${{\bf{V}}^{\left( {b + 1} \right)}}$ and $\{ {\bf{q}}_1^{\left( {b + 1} \right)}[ n],$\\
 \quad $\cdots ,{\bf{q}}_{l - 1}^{\left( {b + 1} \right)}[n],{\bf{q}}_{l }^{\left( b \right)}[n], \cdots ,{\bf{q}}_L^{\left( b \right)}[n]\}$, and denote the \\
 \quad solution as ${\bf{q}}_{l}^{\left( {b + 1} \right)}\left[ n \right]$.
 \STATE \textbf{end for}
 \STATE Update $b=b+1$.
 \UNTIL the fractional increase in the objective function value is below a given threshold $ \epsilon_2 $.
 \end{algorithmic}
 \end{algorithm}

\section{Numerical Results}\label{sectionNumericalResult}
 In this section, numerical results are provided to evaluate the performance of the proposed trajectory optimization scheme for UAV swarm enabled AMA system. The channel power gain at the reference distance of ${r_0} = 1$ m is set as ${\beta _0} =  - 61.4$ dB, and the noise power is ${\sigma ^2} =  - 94$ dBm. The minimum distance to avoid the collision among UAVs is ${d_{\min}} = 5$ m, and the maximum UAV speed is $V_{\max} = 30$ m/s. Each UAV returns to its initial position at the end of operation time. Unless otherwise stated, $K$ UEs are uniformly distributed in a rectangular area centered at the origin, with its length and width being $D_x = 160$ m and $D_y = 120$ m, respectively. The minimum altitude of all UAVs is $H = 100$ m.

 \subsection{Single UE and Two UEs}

 For the case of single UE, its location is ${\bf{w}} = {\left[ {0,0,0} \right]^T}$ m. Under the conditions when $N=1$ or $N$ is sufficiently large, the UAV swarm trajectory optimization is simplified to placement optimization. Fig.~\ref{fig:optimizedPlacement} shows the top view of the optimized UAV swarm placement positions, and the optimized $z$-coordinate values of all UAVs are both equal to $H = 100$ m for the cases of $L = 7$ and $L = 31$. The transmit power is $P= 30$ dBm. It is observed that when $L = 7$, the UAVs form a regular hexagon in the $x$-$y$ plane, where the UAVs are located at the center and vertices of the regular hexagon, and the distance between neighboring UAVs is $d_{\min} = 5$ m. Moreover, for the case of $L = 31$ UAVs, the UAVs form a cellular-like array architecture. The observation provides a strategy for UAV swarm placement in the single UE case, i.e., deploying them at the center and vertices of multiple regular hexagons in the form of cellular array architecture.


 \begin{figure}
 \centering
 \subfigure[$L= 7$]{
 \begin{minipage}[t]{0.5\textwidth}
 \centering
 \centerline{\includegraphics[width=2.66in,height=2.0in]{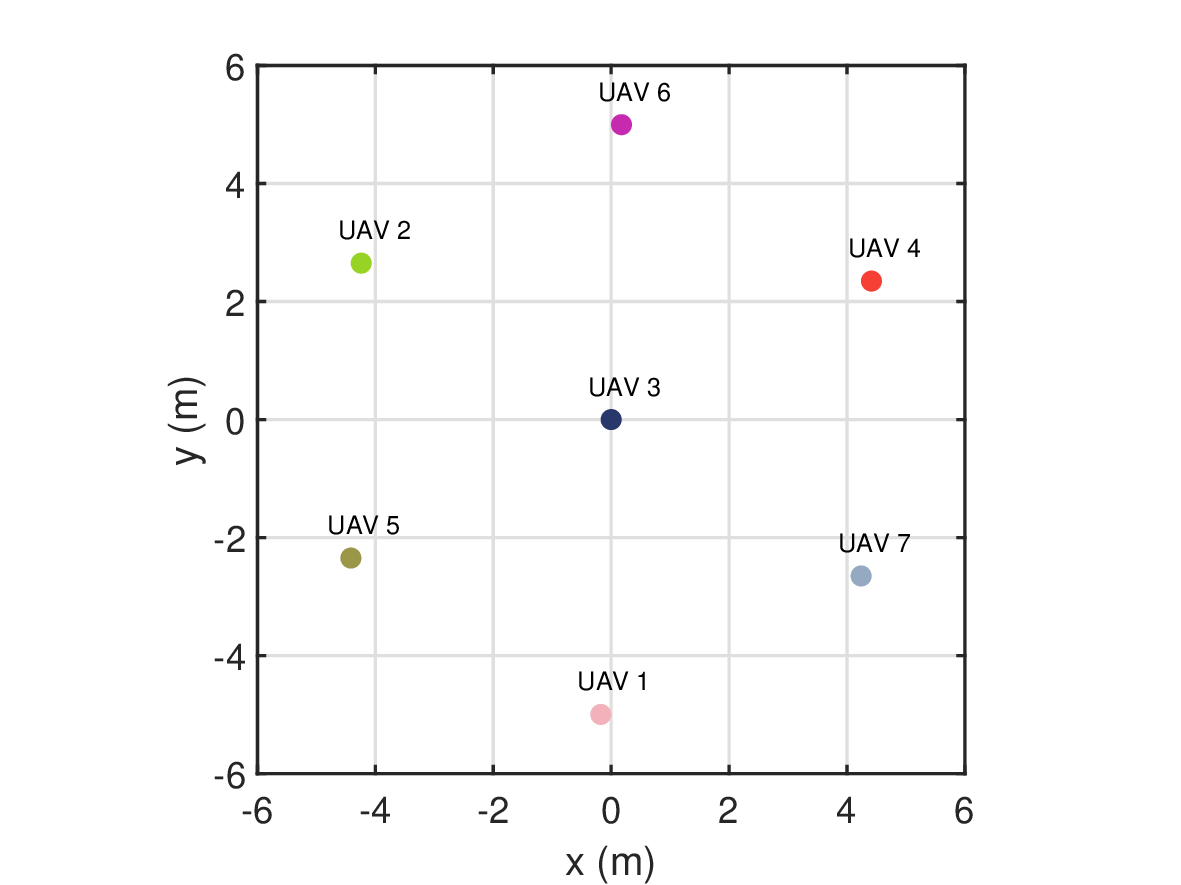}}
 \end{minipage}
 }
 \subfigure[$L= 31$]{
 \begin{minipage}[t]{0.5\textwidth}
 \centering
 \centerline{\includegraphics[width=2.66in,height=2.0in]{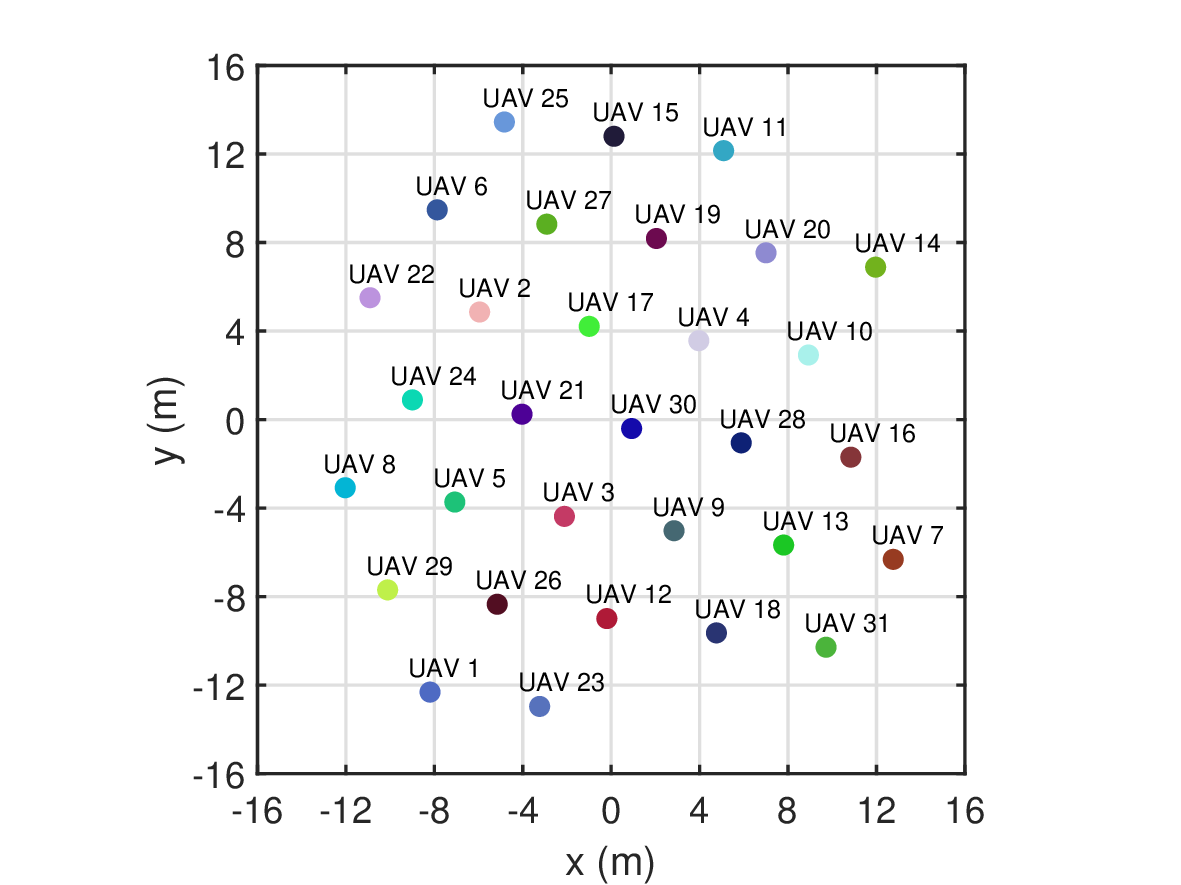}}
 \end{minipage}
 }
 \caption{The optimized UAV swarm placement positions (top view).}
 \label{fig:optimizedPlacement}
 \end{figure}

 \begin{figure}[!t]
 \centering
 \centerline{\includegraphics[width=3.0in,height=2.25in]{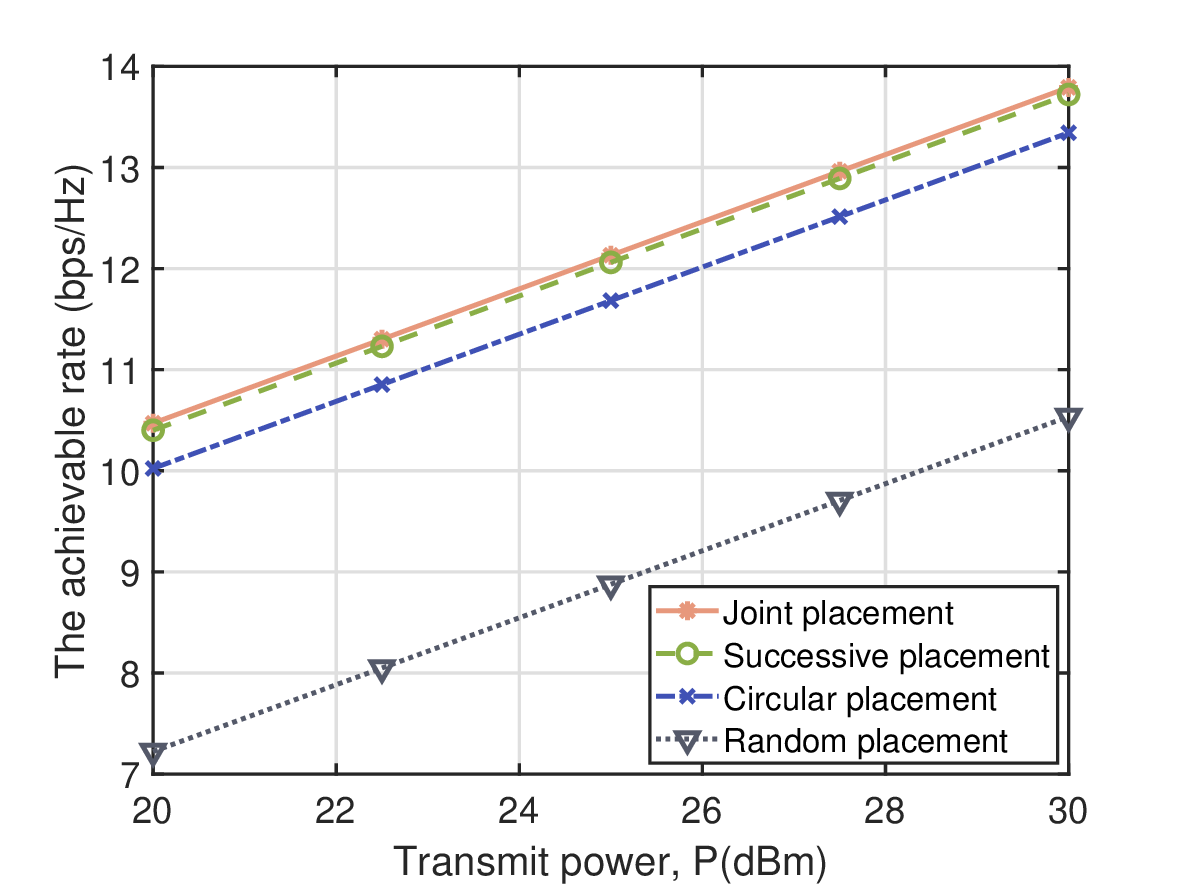}}
 \caption{The achievable rate versus the transmit power for a single UE.}
 \label{fig:achievableRateVersusTransmitPowerSingleUE}
 \end{figure}

 Fig.~\ref{fig:achievableRateVersusTransmitPowerSingleUE} shows the achievable rate versus the transmit power for a single UE. The number of UAVs is $L= 80$. For comparison, we consider the following two benchmark schemes: (1) Circular placement: the UAV swarm is arranged in a circular formation, and adjacent UAVs are separated by $d_{\min}$. The placement position of UAV $l$, $1\le l \le L$, is given by
 \begin{equation}\label{circularPlacement}
 {{\bf{q}}_l} = {\left[ {R\cos \left( {\frac{{2\pi }}{L}\left( {l - 1} \right)} \right),R\sin \left( {\frac{{2\pi }}{L}\left( {l - 1} \right)} \right),H} \right]^T},
 \end{equation}
 where $R = \frac{{{d_{\min }}}}{{2\sin \left( {\pi /L} \right)}}$ is the radius of the circular array; (2) Random placement: the UAV swarm is arranged in a random formation that satisfies the collision avoidance and minimum altitude constraints in \eqref{OptimizationProblemSingleUEJoint}. It is observed that both the joint placement scheme from \eqref{OptimizationProblemSingleUEJoint} and successive placement scheme from \eqref{OptimizationProblemSingleUESuccessive} outperform the two benchmark schemes, as expected. Moreover, the successive placement scheme achieves a comparable performance to the joint placement scheme, while maintaining a lower computational complexity, which verifies its effectiveness in balancing the performance and complexity.

 \begin{figure}[!t]
 \centering
 \centerline{\includegraphics[width=3.0in,height=2.25in]{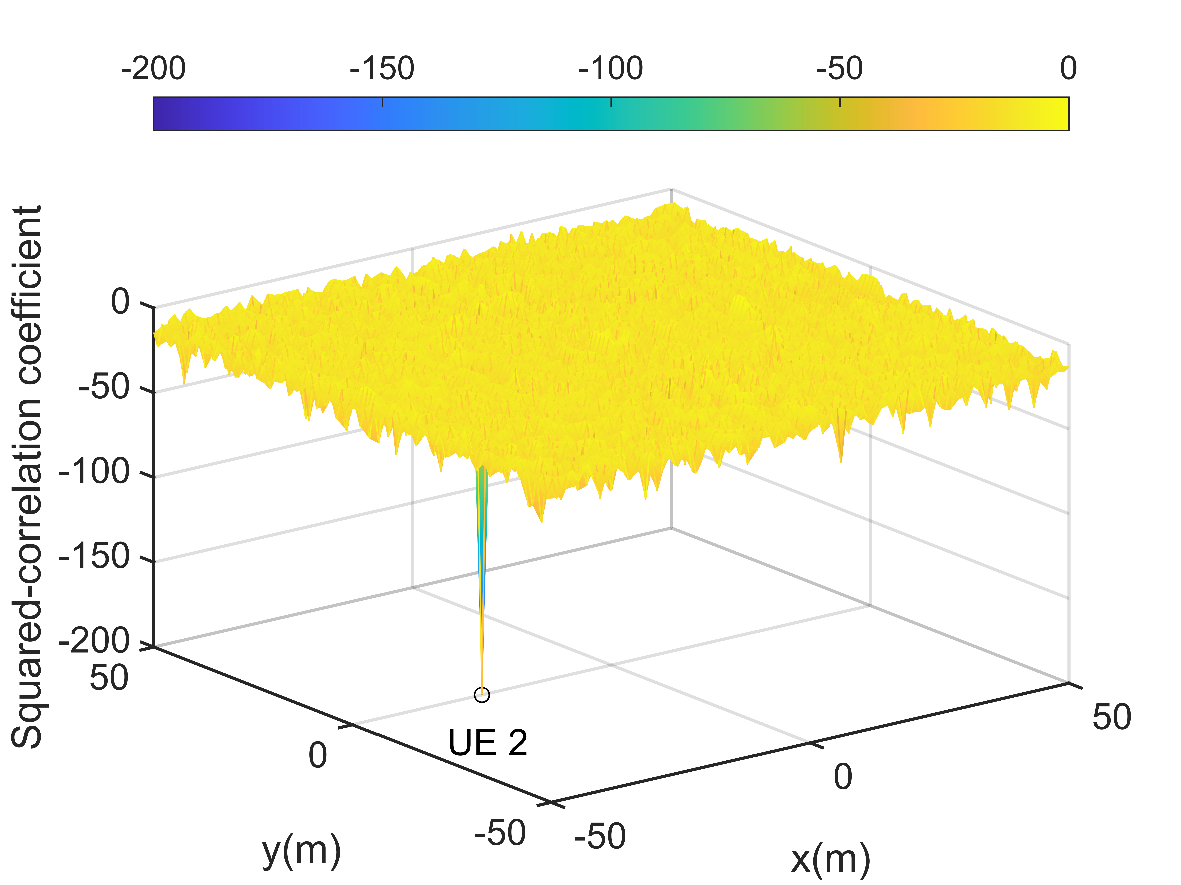}}
 \caption{Squared-correlation coefficient of UE 1 versus arbitrary location.}
 \label{fig:squaredCorrelationCoefficient}
 \end{figure}

 Next, for the case of two UEs, their locations are ${\bf{w}}_1 = {\left[ {25,0,0} \right]^T}$ m and ${\bf{w}}_2 = {\left[ {-25,0,0} \right]^T}$ m, respectively, and the transmit power is $P_k =30$ dBm, $k=1,2$. The number of UAVs is $L = 16$. Fig.~\ref{fig:squaredCorrelationCoefficient} shows the squared-correlation coefficient of UE 1 versus arbitrary location, where the squared-correlation coefficient below $-200$ dB is truncated to $-200$ dB for ease of presentation. It is observed that the proposed placement scheme tailored to two UEs in Section~\ref{subSectionTwoUEs} achieves a zero squared-correlation coefficient between the channels of UE 1 and UE 2, which implies that UE 1 is free from the IUI caused by UE 2, while incurring no SNR loss (as can be seen in \eqref{SINRTwoUE}). Moreover, the squared-correlation coefficient of UE 2 versus arbitrary location can be showed and is omitted for brevity. Similarly, UE 2 can also achieve the SNR ${{\bar P}_2}{\left\| {{{\bf{h}}_2}} \right\|^2}$ as if there is no IUI, thanks to the flexible placement position adjustment of UAV swarm to orthogonalize the channels of the two desired UEs.

 \subsection{Arbitrary Number of UEs}

 In this subsection, we consider the arbitrary number of UEs, with $L =4$ and $K =4$. The initial and final positions of UAVs are ${{\bf{q}}_{I,1}} = {{\bf{q}}_{F,1}} = {\left[ {80,60,100} \right]^T}$ m, ${{\bf{q}}_{I,2}} = {{\bf{q}}_{F,2}} = {\left[ {-80,60,100} \right]^T}$ m, ${{\bf{q}}_{I,3}} = {{\bf{q}}_{F,3}} = {\left[ {-80,-60,100} \right]^T}$ m, and ${{\bf{q}}_{I,4}} = {{\bf{q}}_{F,4}} = {\left[ {80,-60,100} \right]^T}$ m, respectively. 
 Fig.~\ref{fig:illustrationAMAMultiUECommunication} illustrates the UAV swarm enabled AMA multi-UE communication, with the locations of four UEs given by ${{\bf{w}}_1} = {\left[ {40,30,0} \right]^T}$ m, ${{\bf{w}}_2} = {\left[ {-40,30,0} \right]^T}$ m, ${{\bf{w}}_3} = {\left[ {-40,-30,0} \right]^T}$ m, and ${{\bf{w}}_4} = {\left[ {40,-30,0} \right]^T}$ m, respectively. It is observed from Fig.~\ref{fig:illustrationAMAMultiUECommunication}(a) that the UAVs firstly move at a high speed, as reflected by a larger spacing between adjacent trajectory points, and then slightly adjust their trajectories around the geometric center of the UEs. This is attributed to the fact that the UAV swarm can achieve a larger communication rate at these positions, as can be seen in Fig.~\ref{fig:illustrationAMAMultiUECommunication}(b). Thus, a higher UAV speed at the beginning is preferred to exploit favorable channel conditions for improving the minimum average communication rate. Moreover, the UAVs fastly return to the final positions at the end of operation time, and a substantial degradation in the achievable rates for all UEs is observed.


 \begin{figure}
 \centering
 \subfigure[UAV swarm trajectory, where solid dots denote the initial/final\newline positions of UAVs. The positions of UAVs' trajectories at each\newline time slot are marked by the symbol '$\triangle$']{
 \begin{minipage}[t]{0.5\textwidth}
 \centering
 \centerline{\includegraphics[width=3.0in,height=2.25in]{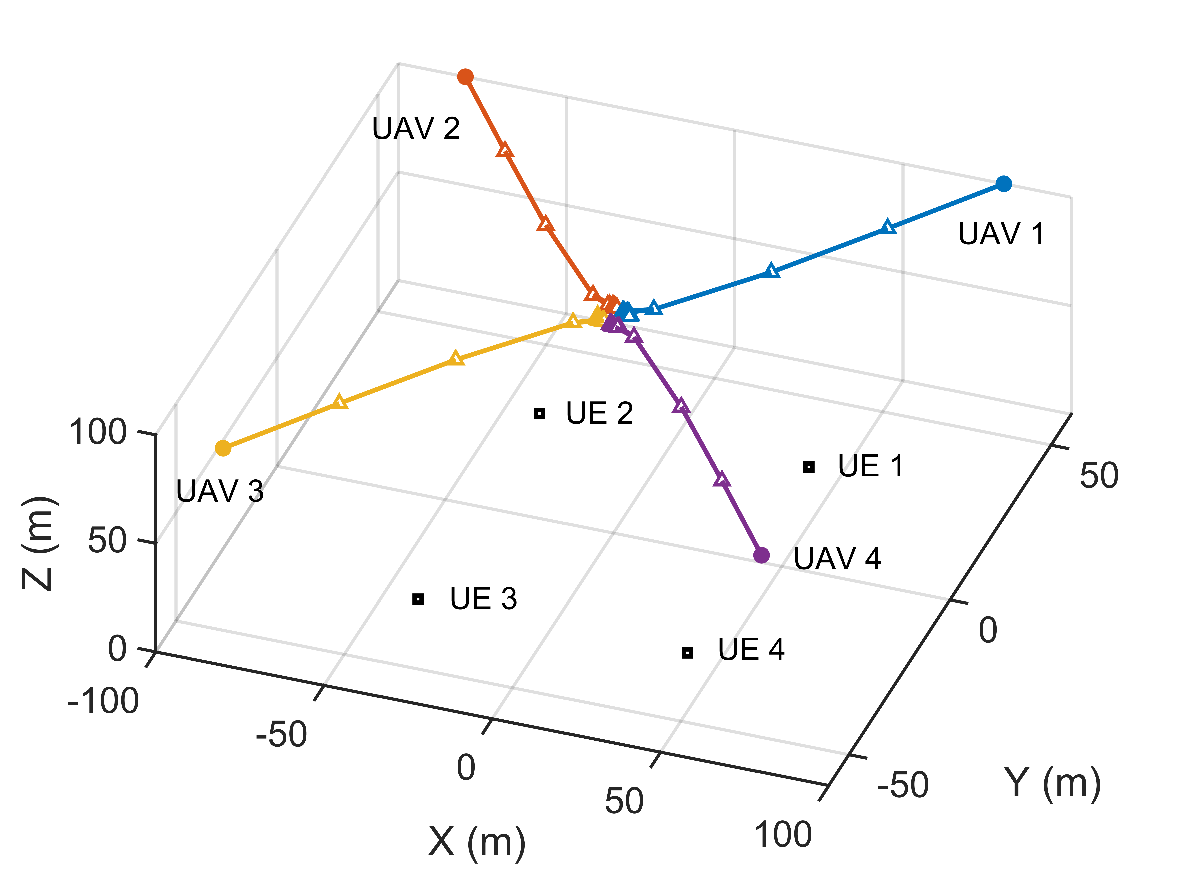}}
 \end{minipage}
 }
 \subfigure[Achievable rate versus the time slot for different UEs]{
 \begin{minipage}[t]{0.5\textwidth}
 \centering
 \centerline{\includegraphics[width=3.0in,height=2.25in]{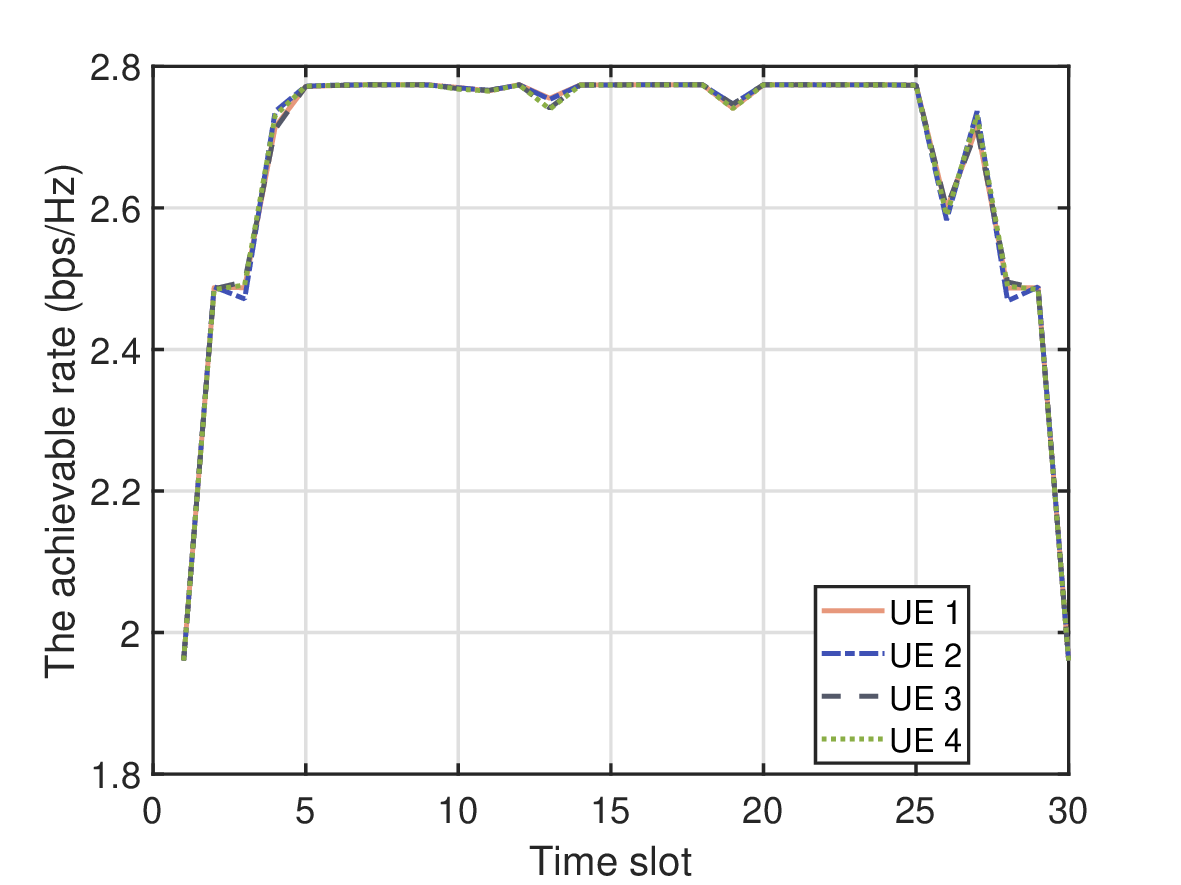}}
 \end{minipage}
 }
 \caption{An illustration of UAV swarm enabled AMA multi-UE communication.}
 \label{fig:illustrationAMAMultiUECommunication}
 \end{figure}

  \begin{figure}[!t]
 \centering
 \centerline{\includegraphics[width=3.0in,height=2.25in]{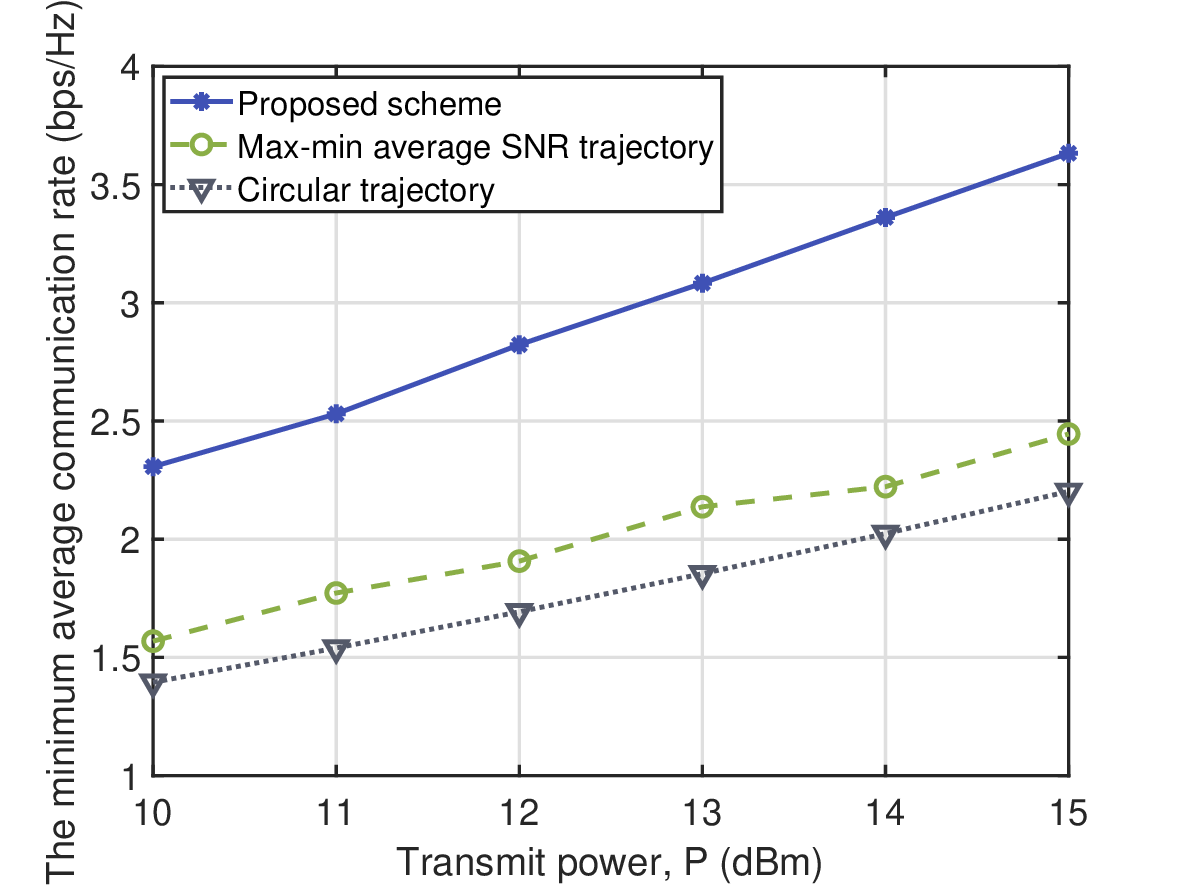}}
 \caption{The minimum average communication rate versus the transmit power of each UE.}
 \label{fig:achievableRateVersusTransmitPower}
 \end{figure}

 Fig.~\ref{fig:achievableRateVersusTransmitPower} shows the minimum average communication rate versus the transmit power of each UE. As a comparison, the following two benchmark schemes are considered: (1) Circular trajectory: the UAV swarm is arranged in a circular formation and follows a circular trajectory. The trajectory of UAV $l$ is given by
 \begin{equation}\label{circularTrajectory}
 \begin{aligned}
 {{\bf{q}}_l}\left[ n \right] &= \left[ {R\cos \left( {\frac{{2\pi }}{{N - 1}}\left( {n - 1} \right) + {\theta _l}} \right),} \right.\\
 &\ \ \ \ \ \ \ \ {\left. {R\sin \left( {\frac{{2\pi }}{{N - 1}}\left( {n - 1} \right) + {\theta_l}} \right),H} \right]^T},
 \end{aligned}
 \end{equation}
 where $R$ is the radius and given below \eqref{circularPlacement}, and ${\theta _l} = \frac{{2\pi }}{L}\left( {l - 1} \right)$ denotes the initial angle (in radians) of UAV $l$ relative to $x$-axis; (2) Max-min average SNR trajectory: the UAV swarm trajectory is obtained in Section~\ref{subSectionUAVTrajectoryInitialization}. It is observed that the proposed scheme yields a better performance than the benchmark scheme based on the max-min average SNR trajectory that only considers its impact on the channel amplitude. This demonstrates the necessity of taking into account the impact of trajectory on both the channel amplitude and phase in the UAV swarm enabled near-field AMA communication. It is also observed that the proposed scheme significantly outperforms the benchmarking circular trajectory, since the former strikes a good balance between the desired signal enhancement and IUI mitigation during the operation time, thanks to the additional design DoF of UAV swarm trajectory.

 \begin{figure}[!t]
 \centering
 \centerline{\includegraphics[width=3.0in,height=2.25in]{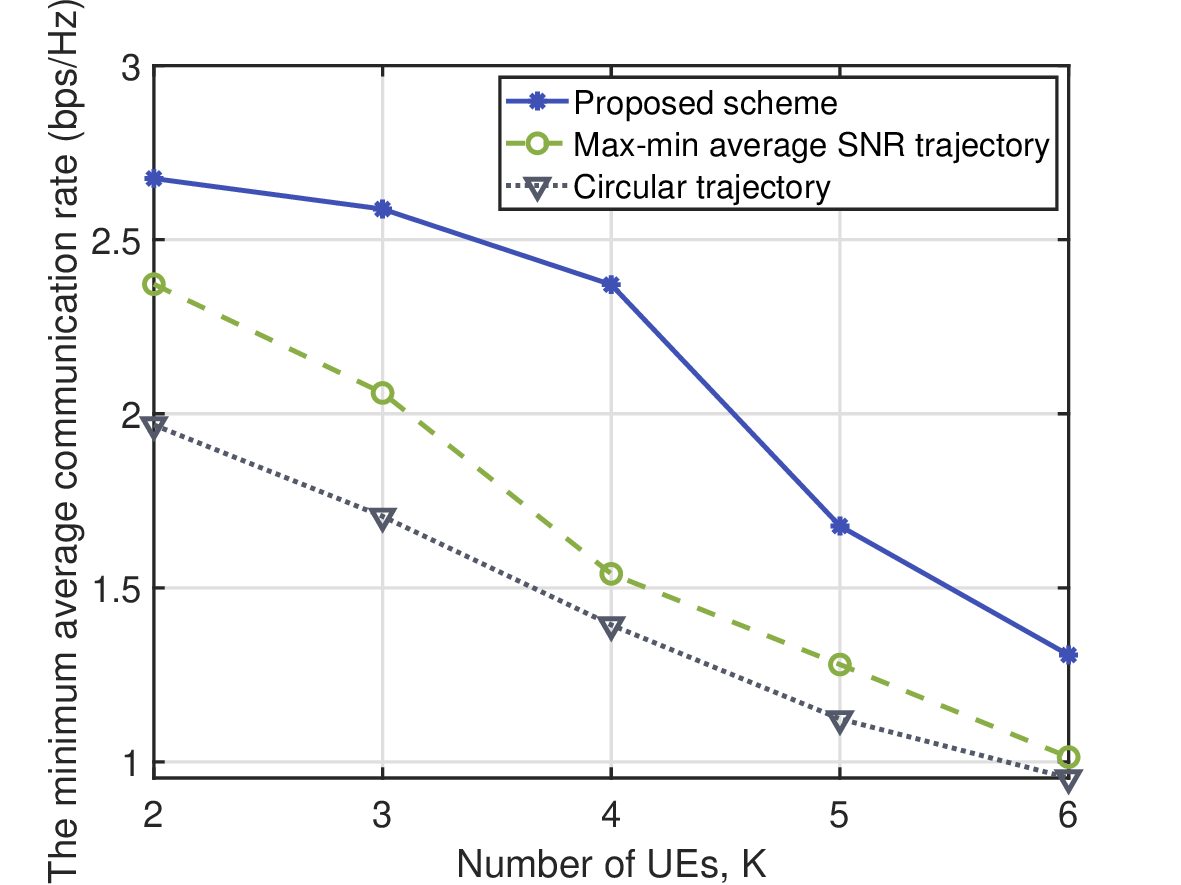}}
 \caption{The minimum average communication rate versus the number of UEs.}
 \label{fig:achievableRateVersusUENumber}
 \end{figure}

 Fig.~\ref{fig:achievableRateVersusUENumber} shows the minimum average communication rate versus the number of UEs $K$. As can be observed, the proposed scheme is superior to the two benchmark schemes. This is expected since the proposed scheme can leverage the channel variations by dynamically adjusting the UAV swarm trajectory, so as to improve the communication performance during the operation time. By contrast, the benchmark scheme based on the max-min average SNR trajectory only optimizes the desired signal power, while neglecting both IUI and the impact of UAV swarm trajectory on the channel phase, and such a limitation leads to an inferior performance than the proposed scheme. Moreover, the minimum average communication rates of all the schemes exhibit a decreasing trend as the number of UEs increases. This degradation occurs because the IUI becomes more pronounced as $K$ increases, and even the proposed scheme cannot ideally mitigate the severe IUI due to the limited spatial DoF. Fortunately, in contrast to conventional cellular networks with the fixed antenna configuration, additional UAVs can be flexibly deployed into the AMA system for an enhanced IUI mitigation capability.


%
%
%
%
\section{Conclusion}\label{sectionConclusion}
 This paper investigated the performance of near-field AMA system enabled by UAV swarm cooperation. To maximize the minimum average communication rate over UEs, we jointly optimized the 3D UAV swarm trajectory and receive beamforming for UEs, subject to practical constraints on UAVs including maximum speed, collision avoidance, as well as initial and final positions. To draw useful insights, the special cases of single UE and two UEs were studied, respectively. In particular, it was shown that an IUI-free communication for two UEs can be achieved by symmetrically placing an even number of UAVs along a hyperbola, with two UEs located at its foci. Furthermore, an alternating optimization was proposed to efficiently tackle the non-convex optimization problem for arbitrary number of UEs. Numerical results demonstrated that the proposed scheme significantly outperforms various benchmark schemes.

\begin{appendices}
\section{Proof of Proposition~\ref{UAVPlacementL2Proposition}}\label{proofofUAVPlacementL2Proposition}
It is observed from \eqref{optimizationProblemL2} that we can set $y_1 = y_2 = 0$ without loss of optimality. Let ${\left[ {x,0,H} \right]^T}$ and ${\left[ {x + {d_{\min }} ,0,H} \right]^T}$ be the positions of two UAVs, and the objective function of \eqref{optimizationProblemL2} can be equivalently expressed as
 \begin{equation}
 h\left( x \right) = \frac{1}{{{x^2} + {H^2}}} + \frac{1}{{{{\left( {x + {d_{\min }}} \right)}^2} + {H^2}}}.
 \end{equation}
 By letting the first-order derivative of $h\left( x \right)$ be zero, and after some manipulations, we have
 \begin{equation}\label{zeroFirstOrderDerivative}
 \begin{aligned}
 y\left[ {{y^4} + 2d_{\min }^2\left( {{\zeta ^2} + \frac{1}{4}} \right){y^2} + d_{\min }^4\left( {{\zeta ^2} + \frac{1}{4}} \right)\left( {{\zeta ^2} - \frac{3}{4}} \right)} \right] = 0,
 \end{aligned}
 \end{equation}
 where $y \triangleq x + {d_{\min }}/2$.

 \emph{Case 1:} $\zeta  > \sqrt 3 /2$. In this case, the terms in brackets of \eqref{zeroFirstOrderDerivative} is larger than zero. The optimal solution is $y^{\star} = 0$, i.e, $x^{\star} = -d_{\min}/2$, and the two UAVs are placed symmetrically w.r.t. the UE's horizontal position.

 \emph{Case 2:} $\zeta  \le \sqrt 3 /2$. In this case, by letting the terms in brackets of \eqref{zeroFirstOrderDerivative} be zero, we have
 \begin{equation}
 y =  \pm {d_{\min }}\sqrt {\sqrt {{\zeta ^2} + \frac{1}{4}}  - \left( {{\zeta ^2} + \frac{1}{4}} \right)}.
 \end{equation}
 It can be verified that the maximum objective function is achieved at the two points, which corresponds to
 \begin{equation}
 {x^ \star } =  - \frac{{{d_{\min }}}}{2} \pm {d_{\min }}\sqrt {\sqrt {{\zeta ^2} + \frac{1}{4}}  - \left( {{\zeta ^2} + \frac{1}{4}} \right)} .
 \end{equation}
 By integrating the two cases, the proof of Proposition~\ref{UAVPlacementL2Proposition} is thus completed.

\section{Proof of Proposition~\ref{zerocorrelationCoefficientProposition}}\label{proofofzerocorrelationCoefficientProposition}
 To show Proposition~\ref{zerocorrelationCoefficientProposition}, the $L$ (even) UAVs are first matched in pairs. Specifically, for given UAV $a$'s placement position, ${{\bf{q}}_a} = {\left[ {{x_a},0,{z_a}} \right]^T}$, $a \in \left\{ {1, \cdots ,L/2} \right\}$, the position of its matched UAV $a'$ is set as ${{\bf{q}}_{a'}} = {\left[ { - {x_a},0,{z_a}} \right]^T}$, with $a' = a + L/2$. It then follows that
 ${r_{1,a}} = {r_{2,a'}} = \sqrt {{{\left( {{x_a} - X} \right)}^2} + z_a^2} $, and ${r_{2,a}} = {r_{1,a'}} = \sqrt {{{\left( {{x_a} + X} \right)}^2}  + {z_a^2} }$. As such, ${\rho _{k,k'}}$ can be expressed as
 \begin{equation}\label{correlationCoefficientPair}
 \begin{aligned}
 {\rho _{k,k'}} &= \frac{{{{\left| {\sum\limits_{a = 1}^{L/2} {\frac{{{r_k}{r_{k'}}}}{{{r_{k,a}}{r_{k',a}}}}\left( {{e^{{\rm{j}}\frac{{2\pi }}{\lambda }\left( {{r_{k,a}} - {r_{k',a}}} \right)}} + {e^{ - {\rm{j}}\frac{{2\pi }}{\lambda }\left( {{r_{k,a}} - {r_{k',a}}} \right)}}} \right)} } \right|}^2}}}{{{{\left\| {{{\bf{a}}_k}} \right\|}^2}{{\left\| {{{\bf{a}}_{k'}}} \right\|}^2}}} \\
 &= \frac{{{{\left| {\sum\limits_{a = 1}^{L/2} {\frac{{2{r_k}{r_{k'}}}}{{{r_{k,a}}{r_{k',a}}}}\cos \left( {\frac{{2\pi }}{\lambda }\left( {{r_{k,a}} - {r_{k',a}}} \right)} \right)} } \right|}^2}}}{{{{\left\| {{{\bf{a}}_k}} \right\|}^2}{{\left\| {{{\bf{a}}_{k'}}} \right\|}^2}}}.
 \end{aligned}
 \end{equation}

 In particular, \eqref{correlationCoefficientPair} is equal to zero if the following condition is satisfied
 \begin{equation}\label{sufficientCondition}
 \frac{{2\pi }}{\lambda }\left( {{r_{k,a}} - {r_{k',a}}} \right) = \frac{\pi }{2} + \nu \pi,\ \nu  \in {\mathbb Z},\ \forall a,
 \end{equation}
 which can be equivalently written as
 \begin{equation}\label{sufficientCondition2}
 {r_{k,a}} - {r_{k',a}} = \frac{{\left( {2\nu  + 1} \right)\lambda }}{4},\ \nu  \in {\mathbb Z},\ \forall a.
 \end{equation}
 A closer look at \eqref{sufficientCondition2} shows that UAV $l$ is in fact located on a hyperbola, whose foci are the locations of two UEs. Moreover, the hyperbolic equation is
 \begin{equation}
 \frac{{x_a^2}}{{{{\left( {\frac{{\left( {2\nu  + 1} \right)\lambda }}{8}} \right)}^2}}} - \frac{{z_a^2}}{{{X^2} - {{\left( {\frac{{\left( {2\nu  + 1} \right)\lambda }}{8}} \right)}^2}}} = 1.
 \end{equation}
 Then, we have
 \begin{equation}\label{hyperbolicEquation}
 {x_a} =  \pm \frac{{\left( {2\nu  + 1} \right)\lambda }}{8}\sqrt {\frac{{z_a^2}}{{{X^2} - {{\left( {\frac{{\left( {2\nu  + 1} \right)\lambda }}{8}} \right)}^2}}} + 1}.
 \end{equation}
 This thus completes the proof of Proposition~\ref{zerocorrelationCoefficientProposition}.

\section{Proof of Lemma \ref{quadraticSurrogatelemma}}\label{proofOfquadraticSurrogatelemma}
 With ${f_{k,i,l}}\left[ n \right]$ given in \eqref{fkil}, its gradient over ${{\bf{q}}_l}\left[ n \right]$, denoted as $\nabla {f_{k,i,l}}\left[ n \right] = {\left[ {\frac{{\partial {f_{k,i,l}}\left[ n \right]}}{{\partial {x_l}\left[ n \right]}},\frac{{\partial {f_{k,i,l}}\left[ n \right]}}{{\partial {y_l}\left[ n \right]}},\frac{{\partial {f_{k,i,l}}\left[ n \right]}}{{\partial {z_l}\left[ n \right]}}} \right]^T}$, is given by \eqref{firstOrderDerivative}, where ${\zeta _{k,i,l,l'}}\left[ n \right] \triangleq \frac{{2{{\left| {{\alpha _i}} \right|}^2}\left| {{V_{k,l,l'}}\left[ n \right]} \right|r_i^2}}{{{{\hat r}_{i,l}}\left[ n \right]{{\hat r}_{i,l'}}\left[ n \right]}}$.

 \newcounter{mytempeqncnt1}
 \begin{figure*}
 \normalsize
 \setcounter{mytempeqncnt1}{\value{equation}}
 \begin{subequations}\label{firstOrderDerivative}
 \begin{align}
 \frac{{\partial {f_{k,i,l}}\left[ n \right]}}{{\partial {x_l}\left[ n \right]}} =  - \frac{{2\pi }}{\lambda }\sum\limits_{l' = 1,l' \ne l}^L {{\zeta _{k,i,l,l'}}\left[ n \right]\frac{{{x_l}\left[ n \right] - {x_i}}}{{\left\| {{{\bf{q}}_l}\left[ n \right] - {{\bf{w}}_i}} \right\|}}} \sin \left( {\frac{{2\pi }}{\lambda }\left( {\left\| {{{\bf{q}}_l}\left[ n \right] - {{\bf{w}}_i}} \right\| - \left\| {{{\bf{q}}_{l'}}\left[ n \right] - {{\bf{w}}_i}} \right\|} \right) + \angle {V_{k,l,l'}}\left[ n \right]} \right),\\
 \frac{{\partial {f_{k,i,l}}\left[ n \right]}}{{\partial {y_l}\left[ n \right]}} =  - \frac{{2\pi }}{\lambda }\sum\limits_{l' = 1,l' \ne l}^L {{\zeta _{k,i,l,l'}}\left[ n \right]\frac{{{y_l}\left[ n \right] - {y_i}}}{{\left\| {{{\bf{q}}_l}\left[ n \right] - {{\bf{w}}_i}} \right\|}}} \sin \left( {\frac{{2\pi }}{\lambda }\left( {\left\| {{{\bf{q}}_l}\left[ n \right] - {{\bf{w}}_i}} \right\| - \left\| {{{\bf{q}}_{l'}}\left[ n \right] - {{\bf{w}}_i}} \right\|} \right) + \angle {V_{k,l,l'}}\left[ n \right]} \right),\\
 \frac{{\partial {f_{k,i,l}}\left[ n \right]}}{{\partial {z_l}\left[ n \right]}} =  - \frac{{2\pi }}{\lambda }\sum\limits_{l' = 1,l' \ne l}^L {{\zeta _{k,i,l,l'}}\left[ n \right]\frac{{{z_l}\left[ n \right] - {z_i}}}{{\left\| {{{\bf{q}}_l}\left[ n \right] - {{\bf{w}}_i}} \right\|}}} \sin \left( {\frac{{2\pi }}{\lambda }\left( {\left\| {{{\bf{q}}_l}\left[ n \right] - {{\bf{w}}_i}} \right\| - \left\| {{{\bf{q}}_{l'}}\left[ n \right] - {{\bf{w}}_i}} \right\|} \right) + \angle {V_{k,l,l'}}\left[ n \right]} \right).
 \end{align}
 \end{subequations}
 \hrulefill
 \end{figure*}
 Then, the Hessian matrix of ${f_{k,i,l}}\left[ n \right]$ over ${{\bf{q}}_l}\left[ n \right]$ is
 \begin{equation}\label{HessianMatrix}
 {\nabla ^2}{f_{k,i,l}}\left[ n \right]  = \left[ {\begin{array}{*{20}{c}}
 {\frac{{{\partial ^2}{f_{k,i,l}}\left[ n \right]}}{{\partial {x_l}\left[ n \right]\partial {x_l}\left[ n \right]}}}&{\frac{{{\partial ^2}{f_{k,i,l}}\left[ n \right]}}{{\partial {x_l}\left[ n \right]\partial {y_l}\left[ n \right]}}}&{\frac{{{\partial ^2}{f_{k,i,l}}\left[ n \right]}}{{\partial {x_l}\left[ n \right]\partial {z_l}\left[ n \right]}}}\\
 {\frac{{{\partial ^2}{f_{k,i,l}}\left[ n \right]}}{{\partial {y_l}\left[ n \right]\partial {x_l}\left[ n \right]}}}&{\frac{{{\partial ^2}{f_{k,i,l}}\left[ n \right]}}{{\partial {y_l}\left[ n \right]\partial {y_l}\left[ n \right]}}}&{\frac{{{\partial ^2}{f_{k,i,l}}\left[ n \right]}}{{\partial {y_l}\left[ n \right]\partial {z_l}\left[ n \right]}}}\\
 {\frac{{{\partial ^2}{f_{k,i,l}}\left[ n \right]}}{{\partial {z_l}\left[ n \right]\partial {x_l}\left[ n \right]}}}&{\frac{{{\partial ^2}{f_{k,i,l}}\left[ n \right]}}{{\partial {z_l}\left[ n \right]\partial {y_l}\left[ n \right]}}}&{\frac{{{\partial ^2}{f_{k,i,l}}\left[ n \right]}}{{\partial {z_l}\left[ n \right]\partial {z_l}\left[ n \right]}}}
 \end{array}} \right],
 \end{equation}
 where the second-order partial derivatives $\frac{{{\partial ^2}{f_{k,i,l}}\left[ n \right]}}{{\partial {x_l}\left[ n \right]\partial {x_l}\left[ n \right]}}$, $\frac{{{\partial ^2}{f_{k,i,l}}\left[ n \right]}}{{\partial {x_l}\left[ n \right]\partial {y_l}\left[ n \right]}}$, and $\frac{{{\partial ^2}{f_{k,i,l}}\left[ n \right]}}{{\partial {x_l}\left[ n \right]\partial {z_l}\left[ n \right]}}$ are given by \eqref{secondOrderDerivative}, and the remaining terms can be similarly obtained, which are omitted for brevity.
 \newcounter{mytempeqncnt2}
 \begin{figure*}
 \normalsize
 \setcounter{mytempeqncnt2}{\value{equation}}
 \begin{subequations}\label{secondOrderDerivative}
 \begin{align}
 &\frac{{{\partial ^2}{f_{k,i,l}}\left[ n \right]}}{{\partial {x_l}\left[ n \right]\partial {x_l}\left[ n \right]}} =  - \frac{{2\pi }}{\lambda }\sum\limits_{l' = 1,l' \ne l}^L {{\zeta _{k,i,l,l'}}\left[ n \right]\left[ {\frac{{\left\| {{{\bf{q}}_l}\left[ n \right] - {{\bf{w}}_i}} \right\| - \frac{{{{\left( {{x_l}\left[ n \right] - {x_i}} \right)}^2}}}{{\left\| {{{\bf{q}}_l}\left[ n \right] - {{\bf{w}}_i}} \right\|}}}}{{{{\left\| {{{\bf{q}}_l}\left[ n \right] - {{\bf{w}}_i}} \right\|}^2}}}\sin \left( {\frac{{2\pi }}{\lambda }\left( {\left\| {{{\bf{q}}_l}\left[ n \right] - {{\bf{w}}_i}} \right\| - \left\| {{{\bf{q}}_{l'}}\left[ n \right] - {{\bf{w}}_i}} \right\|} \right) + \angle {V_{k,l,l'}}\left[ n \right]} \right)} \right.}\notag \\
 &\ \ \ \ \ \ \ \ \ \ \ \ \ \ \ \ \ \ \ \ \ \ \ \ \ \ \ \ \ \left. { + \frac{{2\pi }}{\lambda }\frac{{{{\left( {{x_l}\left[ n \right] - {x_i}} \right)}^2}}}{{{{\left\| {{{\bf{q}}_l}\left[ n \right] - {{\bf{w}}_i}} \right\|}^2}}}\cos \left( {\frac{{2\pi }}{\lambda }\left( {\left\| {{{\bf{q}}_l}\left[ n \right] - {{\bf{w}}_i}} \right\| - \left\| {{{\bf{q}}_{l'}}\left[ n \right] - {{\bf{w}}_i}} \right\|} \right) + \angle {V_{k,l,l'}}\left[ n \right]} \right)} \right],\\
 &\frac{{{\partial ^2}{f_{k,i,l}}\left[ n \right]}}{{\partial {x_l}\left[ n \right]\partial {y_l}\left[ n \right]}} =  - \frac{{2\pi }}{\lambda }\sum\limits_{l' = 1,l' \ne l}^L {{\zeta _{k,i,l,l'}}\left[ n \right]\left[ { - \frac{{\left( {{x_l}\left[ n \right] - {x_i}} \right)\left( {{y_l}\left[ n \right] - {y_i}} \right)}}{{{{\left\| {{{\bf{q}}_l}\left[ n \right] - {{\bf{w}}_i}} \right\|}^3}}}\sin \left( {\frac{{2\pi }}{\lambda }\left( {\left\| {{{\bf{q}}_l}\left[ n \right] - {{\bf{w}}_i}} \right\| - \left\| {{{\bf{q}}_{l'}}\left[ n \right] - {{\bf{w}}_i}} \right\|} \right) + \angle {V_{k,l,l'}}\left[ n \right]} \right)} \right.}\notag \\
 &\ \ \ \ \ \ \ \ \ \ \ \ \ \ \ \ \ \ \ \ \ \ \ \ \ \ \ \ \ \left. { + \frac{{2\pi }}{\lambda }\frac{{\left( {{x_l}\left[ n \right] - {x_i}} \right)\left( {{y_l}\left[ n \right] - {y_i}} \right)}}{{{{\left\| {{{\bf{q}}_l}\left[ n \right] - {{\bf{w}}_i}} \right\|}^2}}}\cos \left( {\frac{{2\pi }}{\lambda }\left( {\left\| {{{\bf{q}}_l}\left[ n \right] - {{\bf{w}}_i}} \right\| - \left\| {{{\bf{q}}_{l'}}\left[ n \right] - {{\bf{w}}_i}} \right\|} \right) + \angle {V_{k,l,l'}}\left[ n \right]} \right)} \right],\\
 &\frac{{{\partial ^2}{f_{k,i,l}}\left[ n \right]}}{{\partial {x_l}\left[ n \right]\partial {z_l}\left[ n \right]}} =  - \frac{{2\pi }}{\lambda }\sum\limits_{l' = 1,l' \ne l}^L {{\zeta _{k,i,l,l'}}\left[ n \right]\left[ { - \frac{{\left( {{x_l}\left[ n \right] - {x_i}} \right)\left( {{z_l}\left[ n \right] - {z_i}} \right)}}{{{{\left\| {{{\bf{q}}_l}\left[ n \right] - {{\bf{w}}_i}} \right\|}^3}}}\sin \left( {\frac{{2\pi }}{\lambda }\left( {\left\| {{{\bf{q}}_l}\left[ n \right] - {{\bf{w}}_i}} \right\| - \left\| {{{\bf{q}}_{l'}}\left[ n \right] - {{\bf{w}}_i}} \right\|} \right) + \angle {V_{k,l,l'}}\left[ n \right]} \right)} \right.}\notag \\
 &\ \ \ \ \ \ \ \ \ \ \ \ \ \ \ \ \ \ \ \ \ \ \ \ \ \ \ \ \ \left. { + \frac{{2\pi }}{\lambda }\frac{{\left( {{x_l}\left[ n \right] - {x_i}} \right)\left( {{z_l}\left[ n \right] - {z_i}} \right)}}{{{{\left\| {{{\bf{q}}_l}\left[ n \right] - {{\bf{w}}_i}} \right\|}^2}}}\cos \left( {\frac{{2\pi }}{\lambda }\left( {\left\| {{{\bf{q}}_l}\left[ n \right] - {{\bf{w}}_i}} \right\| - \left\| {{{\bf{q}}_{l'}}\left[ n \right] - {{\bf{w}}_i}} \right\|} \right) + \angle {V_{k,l,l'}}\left[ n \right]} \right)} \right].
 \end{align}
 \end{subequations}
 \hrulefill
 \end{figure*}

 \newcounter{mytempeqncnt3}
 \begin{figure*}
 \normalsize
 \setcounter{mytempeqncnt3}{\value{equation}}
 \begin{equation}\label{squaredFrobeniusNorm}
 \begin{aligned}
 &\left\| {{\nabla ^2}{f_{k,i,l}}\left[ n \right]} \right\|_F^2 = {\left( {\frac{{{\partial ^2}{f_{k,i,l}}\left[ n \right]}}{{\partial {x_l}\left[ n \right]\partial {x_l}\left[ n \right]}}} \right)^2} +  \cdots  + {\left( {\frac{{{\partial ^2}{f_{k,i,l}}\left[ n \right]}}{{\partial {z_l}\left[ n \right]\partial {z_l}\left[ n \right]}}} \right)^2} \le {\left( {\frac{2{\pi }}{{\lambda}}\sum\limits_{l' = 1,l' \ne l}^L {{\zeta _{k,i,l,l'}}\left[ n \right]} } \right)^2} \times \\
 &\left[ {{{\left( {\frac{{\left\| {{{\bf{q}}_l}\left[ n \right] - {{\bf{w}}_i}} \right\| - \frac{{{{\left( {{x_l}\left[ n \right] - {x_i}} \right)}^2}}}{{\left\| {{{\bf{q}}_l}\left[ n \right] - {{\bf{w}}_i}} \right\|}}}}{{{{\left\| {{{\bf{q}}_l}\left[ n \right] - {{\bf{w}}_i}} \right\|}^2}}}} \right)}^2} + {{\left( {\frac{{2\pi }}{\lambda }\frac{{{{\left( {{x_l}\left[ n \right] - {x_i}} \right)}^2}}}{{{{\left\| {{{\bf{q}}_l}\left[ n \right] - {{\bf{w}}_i}} \right\|}^2}}}} \right)}^2} + {{\left( {\frac{{\left( {{x_l}\left[ n \right] - {x_i}} \right)\left( {{y_l}\left[ n \right] - {y_i}} \right)}}{{{{\left\| {{{\bf{q}}_l}\left[ n \right] - {{\bf{w}}_i}} \right\|}^3}}}} \right)}^2} + } \right. \\
 & {\left( {\frac{{2\pi }}{\lambda }\frac{{\left( {{x_l}\left[ n \right] - {x_i}} \right)\left( {{y_l}\left[ n \right] - {y_i}} \right)}}{{{{\left\| {{{\bf{q}}_l}\left[ n \right] - {{\bf{w}}_i}} \right\|}^2}}}} \right)^2}+{\left( {\frac{{\left( {{x_l}\left[ n \right] - {x_i}} \right)\left( {{z_l}\left[ n \right] - {z_i}} \right)}}{{{{\left\| {{{\bf{q}}_l}\left[ n \right] - {{\bf{w}}_i}} \right\|}^3}}}} \right)^2}  + {\left( {\frac{{2\pi }}{\lambda }\frac{{\left( {{x_l}\left[ n \right] - {x_i}} \right)\left( {{z_l}\left[ n \right] - {z_i}} \right)}}{{{{\left\| {{{\bf{q}}_l}\left[ n \right] - {{\bf{w}}_i}} \right\|}^2}}}} \right)^2} + \\
 &{\left( {\frac{{\left( {{x_l}\left[ n \right] - {x_i}} \right)\left( {{y_l}\left[ n \right] - {y_i}} \right)}}{{{{\left\| {{{\bf{q}}_l}\left[ n \right] - {{\bf{w}}_i}} \right\|}^3}}}} \right)^2} + {\left( {\frac{{2\pi }}{\lambda }\frac{{\left( {{x_l}\left[ n \right] - {x_i}} \right)\left( {{y_l}\left[ n \right] - {y_i}} \right)}}{{{{\left\| {{{\bf{q}}_l}\left[ n \right] - {{\bf{w}}_i}} \right\|}^2}}}} \right)^2}+{\left( {\frac{{\left\| {{{\bf{q}}_l}\left[ n \right] - {{\bf{w}}_i}} \right\| - \frac{{{{\left( {{y_l}\left[ n \right] - {y_i}} \right)}^2}}}{{\left\| {{{\bf{q}}_l}\left[ n \right] - {{\bf{w}}_i}} \right\|}}}}{{{{\left\| {{{\bf{q}}_l}\left[ n \right] - {{\bf{w}}_i}} \right\|}^2}}}} \right)^2} + \\
 &{\left( {\frac{{2\pi }}{\lambda }\frac{{{{\left( {{y_l}\left[ n \right] - {y_i}} \right)}^2}}}{{{{\left\| {{{\bf{q}}_l}\left[ n \right] - {{\bf{w}}_i}} \right\|}^2}}}} \right)^2}+ {\left( {  \frac{{\left( {{y_l}\left[ n \right] - {y_i}} \right)\left( {{z_l}\left[ n \right] - {z_i}} \right)}}{{{{\left\| {{{\bf{q}}_l}\left[ n \right] - {{\bf{w}}_i}} \right\|}^3}}}} \right)^2} + {\left( {\frac{{2\pi }}{\lambda }\frac{{\left( {{y_l}\left[ n \right] - {y_i}} \right)\left( {{z_l}\left[ n \right] - {z_i}} \right)}}{{{{\left\| {{{\bf{q}}_l}\left[ n \right] - {{\bf{w}}_i}} \right\|}^2}}}} \right)^2}+{\left( {  \frac{{\left( {{x_l}\left[ n \right] - {x_i}} \right)\left( {{z_l}\left[ n \right] - {z_i}} \right)}}{{{{\left\| {{{\bf{q}}_l}\left[ n \right] - {{\bf{w}}_i}} \right\|}^3}}}} \right)^2} +  \\
 &{\left( {\frac{{2\pi }}{\lambda }\frac{{\left( {{x_l}\left[ n \right] - {x_i}} \right)\left( {{z_l}\left[ n \right] - {z_i}} \right)}}{{{{\left\| {{{\bf{q}}_l}\left[ n \right] - {{\bf{w}}_i}} \right\|}^2}}}} \right)^2}+{\left( {  \frac{{\left( {{y_l}\left[ n \right] - {y_i}} \right)\left( {{z_l}\left[ n \right] - {z_i}} \right)}}{{{{\left\| {{{\bf{q}}_l}\left[ n \right] - {{\bf{w}}_i}} \right\|}^3}}}} \right)^2} +{\left( {\frac{{2\pi }}{\lambda }\frac{{\left( {{y_l}\left[ n \right] - {y_i}} \right)\left( {{z_l}\left[ n \right] - {z_i}} \right)}}{{{{\left\| {{{\bf{q}}_l}\left[ n \right] - {{\bf{w}}_i}} \right\|}^2}}}} \right)^2}+\\
 &\left. {\left( {\frac{{\left\| {{{\bf{q}}_l}\left[ n \right] - {{\bf{w}}_i}} \right\| - \frac{{{{\left( {{z_l}\left[ n \right] - {z_i}} \right)}^2}}}{{\left\| {{{\bf{q}}_l}\left[ n \right] - {{\bf{w}}_i}} \right\|}}}}{{{{\left\| {{{\bf{q}}_l}\left[ n \right] - {{\bf{w}}_i}} \right\|}^2}}}} \right)^2} + {\left( {\frac{{2\pi }}{\lambda }\frac{{{{\left( {{z_l}\left[ n \right] - {z_i}} \right)}^2}}}{{{{\left\| {{{\bf{q}}_l}\left[ n \right] - {{\bf{w}}_i}} \right\|}^2}}}} \right)^2} \right]= {\left( {\frac{{2\pi }}{\lambda }\sum\limits_{l' = 1,l' \ne l}^L {{\zeta _{k,i,l,l'}}\left[ n \right]} } \right)^2}\left( {\frac{{4{\pi ^2}}}{{{\lambda ^2}}} + \frac{2}{{{{\left\| {{{\bf{q}}_l}\left[ n \right] - {{\bf{w}}_i}} \right\|}^2}}}} \right).
 \end{aligned}
 \end{equation}
 \end{figure*}
  Moreover, an upper bound for the squared Frobenius norm of ${\nabla ^2}{f_{k,i,l}}\left[ n \right]$ is given by \eqref{squaredFrobeniusNorm}. It is worth noting that ${\left\| {{{\bf{q}}_l}\left[ n \right] - {{\bf{w}}_i}} \right\|^2} \ge {H^2}$, and since $\frac{{4{\pi ^2}}}{{{\lambda ^2}}} \gg \frac{2}{{{H^2}}}$, it follows that
 \begin{equation}
 \left\| {{\nabla ^2}{f_{k,i,l}}\left[ n \right]} \right\|_F^2 \le {\left( {\frac{{4{\pi ^2}}}{{{\lambda ^2}}}\sum\limits_{l' = 1,l' \ne l}^L {{\zeta _{k,i,l,l'}}\left[ n \right]} } \right)^2}.
 \end{equation}
 With the properties $\left\| {{\nabla ^2}{f_{k,i,l}}\left[ n \right]} \right\|_2^2 \le \left\| {{\nabla ^2}{f_{k,i,l}}\left[ n \right]} \right\|_F^2$ and ${\nabla ^2}{f_{k,i,l}}\left[ n \right] \preceq {\left\| {{\nabla ^2}{f_{k,i,l}}\left[ n \right]} \right\|_2}{\bf{I}}$, we have ${\nabla ^2}{f_{k,i,l}}\left[ n \right] \preceq {\omega _{k,i,l}}\left[ n \right]{\bf{I}}$, where ${\omega _{k,i,l}}\left[ n \right] \triangleq \frac{{4{\pi ^2}}}{{{\lambda ^2}}}\sum\nolimits_{l' = 1,l' \ne l}^L {{\zeta _{k,i,l,l'}}\left[ n \right]} $. Then, based on the Taylor's theorem \cite{sun2016majorization}, the quadratic surrogate functions serving as the lower and upper bounds of ${f_{k,i,l}}\left[ n \right]$, i.e., \eqref{lowerBoundfkim} and \eqref{upperBoundfkim}, can be respectively obtained. This thus completes the proof of Lemma \ref{quadraticSurrogatelemma}.

\end{appendices}


\bibliographystyle{IEEEtran}
\bibliography{refNearFieldUAVMA}

\end{document}